\documentclass[letterpaper,11pt]{article}
\usepackage{fullpage}
\usepackage{amsmath,amsthm}
\usepackage[numbers,sort]{natbib}

\usepackage{times}
\usepackage{soul}
\usepackage{url}
\usepackage[hidelinks]{hyperref}
\usepackage[utf8]{inputenc}
\usepackage[small]{caption}
\usepackage{graphicx}
\usepackage{amsmath}
\usepackage{amsthm}
\usepackage{booktabs}
\usepackage{amsthm,amsmath,mathtools}
\usepackage{thmtools,thm-restate}
\usepackage{longtable}
\usepackage{subfig}
\usepackage{tikz-cd}
\usetikzlibrary{arrows, automata}
\usepackage{adjustbox}
\usepackage[capitalise,noabbrev]{cleveref}
\usepackage{tabularx}
\usepackage{enumitem}

\usepackage[switch]{lineno}

\usepackage{natbib}

\newtheorem{examplex}{Example}[section]
\newenvironment{example}{\pushQED{\qed}\examplex}{\popQED\endexamplex}
\renewcommand{\>}{\succ}
\newcommand{\DA}{\texttt{DA}}
\newcommand{\NPH}{\textrm{\textup{NP-hard}}}

\newtheorem{lemma}{Lemma}
\newtheorem{corollary}{Corollary}

\Crefname{thm}{Theorem}{Theorems}
\Crefname{remark}{Remark}{Remarks}
\Crefname{examplex}{Example}{Examples}

\title{Fair Stable Matching Meets Correlated Preferences}

\author{
Angelina Brilliantova\\
Rochester Institute of Technology\\
Rochester, NY, USA\\
\texttt{lb9849@rit.edu}
\and
Hadi Hosseini\\
Pennsylvania State University\\
University Park, PA, USA\\
\texttt{hadi@psu.edu}
}

\date{}
\begin{document}

\maketitle


\begin{abstract}
The stable matching problem sets the economic foundation of several practical applications ranging from school choice and medical residency to ridesharing and refugee placement. It is concerned with finding a matching between two disjoint sets of agents wherein no pair of agents prefer each other to their matched partners. The \textit{Deferred Acceptance} (\DA{}) algorithm is an elegant procedure that guarantees a stable matching for any input; however, its outcome may be unfair as it always favors one side by returning a matching that is optimal for one side (say men) and pessimal for the other side (say women). A desirable fairness notion is minimizing the \textit{sex-equality cost}, i.e. the difference between the total rankings of both sides. Computing such stable matchings is a strongly NP-hard problem, which raises the question of what tractable algorithms to adopt in practice. We conduct a series of empirical evaluations on the properties of sex-equal stable matchings when preferences of agents on both sides are correlated. Our empirical results suggest that under correlated preferences, the \DA{} algorithm returns stable matchings with low sex-equality cost, which further confirms its broad use in many practical applications.
\end{abstract}





\section{Introduction}

Matching theory sets the economic foundation for achieving stable allocations through market design. It has shaped the cornerstone of many practical applications ranging from school choice~\citep{APR+05boston,APR05new} and medical residency \cite{R84evolution,roth1999redesign} to refugee placement \cite{delacretaz2019matching,jones2017refugee} and ridesharing \cite{hall2017labor,banerjee2019ride}. 
In its essence, the \textit{stable matching} problem deals with finding a matching between two disjoint sets of agents (colloquially men and women) according to their preferences.
%
The primary objective is to achieve stability, that is, finding a matching between the two sides wherein no pair of agents prefer each other to their matched partners.


Over the past few decades, numerous theoretical breakthroughs and developments were pioneered to study mathematical and axiomatic properties of the stable matching problem \cite{roth1982economics,teo2001gale,demange1987further,GS85some} as well as its computational and algorithmic aspects~\cite{vaish2017manipulating,HUV21accomplice,kato1993complexity,knuth1990stable,shen2018coalition}. 
However as problems become increasingly more challenging, there has been a need for moving from theory about simple markets to more complex settings that account for subtle, but vital, differences in constraints in preference structures or other factors that may pose computational or axiomatic challenges. Hence, as \citet{roth1999redesign} suggested \textit{``as game theory moves from simple conceptual problems to complex design problems, we will need to make more general use of this interaction among theory, computational investigation of market data, and theoretical computation, and that this in turn will produce new problems and directions for traditional theory''}. This doctrine motivated a large body of work in taking empirical or statistical approaches in exploring markets through studying (real or synthetic) data sets or analyzing statistical distributions \citep{ashlagi2017unbalanced,knoblauch2009marriage, boudreau2010marriage,tziavelis20}.

In this vein, we investigate the fairness of stable matchings through empirical simulations to paint a thorough picture of the structure of fair stable solutions in matching markets.

\subsubsection*{\textbf{A motivating insight}}
The Deferred Acceptance algorithm (\DA{})---due to \citet{gale1962college}---provides an elegant solution to the stable matching problem wherein agents from one side (say men) make proposals to the agents from the other side (say women). Each woman tentatively accepts her favorite proposal and rejects the rest.
Despite the popularity and success of the $\DA{}$ algorithm in many real-world matching markets~\citep{roth1999redesign,roth1982economics}, it always favors the proposing side to the receiving side, that is, the \DA{} algorithm always returns a stable matching that is \emph{men-optimal}~\cite{gale1962college} but \emph{women-pessimal}~\citep{mcvitie1971stable}. This unequal treatment of the sides raises critical questions about the \textit{fairness} of the $\DA{}$ algorithm, which may result in (extremely) unequal welfare between both sides.

One of the most prominent and well-studied fairness notions---proposed by \citet{gusfield1989stable}---is \emph{sex-equality} that aims at finding a matching that equalizes the \textit{welfare} of both sides by minimizing the difference between the total rankings of men and women in a stable matching. 
While sex-equal stable matchings always exist, computing one has shown to be strongly \NPH{}~\citep{kato1993complexity}. 
For any instance of the matching problem, there may be an exponential number of stable solutions~\citep{irving1986complexity} and these stable matchings form a \textit{distributive lattice}. Thus, one may hope to find a ``fair'' stable matching that equalizes the welfare of both sides.

On a closer scrutiny, however, we notice that a sex-equal stable matching is highly correlated with the structure of preference lists of each side of the market. This observation suggests that even though the number of stable matchings grows exponentially in general~\citep{irving1986complexity}, surprisingly a sex-equal solution lies at the extreme points of the stable lattice in certain settings. 

\begin{example}\label{example:sex-equal}
Consider the following instance with five men and five women and the following preference lists.
\begin{table}[h!]
\small
   \centering
  \begin{tabularx}{\linewidth}{XXXXXXXXXXXXXXX}
           $m_1\colon$& $w_2^\dagger$& $w_4^*$& $w_5$& $w_1$& $w_3$ &&
           $w_1\colon$&$m_4$& $m_2$& $m_1$& $m_5^*$& $m_3^\dagger$\\
            $m_2\colon$& $w_3^{\dagger}$& $w_2^*$& $w_4$& $w_1$& $w_5$ &&
            $w_2\colon$&$m_2^*$& $m_4$& $m_1^\dagger$& $m_5$& $m_3$\\
            $m_3\colon$&$w_1^\dagger$& $w_5^*$& $w_4$& $w_3$& $w_2$ &&
            $w_3\colon$&$m_4^{*}$& $m_2^{\dagger}$& $m_1$& $m_3$& $m_5$\\
            $m_4\colon$&$w_4^\dagger$& $w_2$& $w_3^*$& $w_1$& $w_5$ &&
            $w_4\colon$&$m_2$& $m_1^*$& $m_4^\dagger$& $m_5$& $m_3$\\
             $m_5\colon$&$w_2$& $w_3$& $w_5^\dagger$& $w_1^*$& $w_4$ &&
             $w_5\colon$&$m_1$&$m_4$& $m_2$& $m_3^*$& $m_5^\dagger$
 \end{tabularx}
\end{table}

  \begin{figure}[t]
    \centering
        \centering
        \begin{tikzpicture}[
                    every edge quotes/.append style={font=\LARGE},
                    every label/.append style={font = \Huge},
                    level 1/.style={sibling distance=20mm},
                    scale=0.4, 
                    every node/.append style={transform shape},
                    > = stealth, 
                    shorten > = 1pt, 
                    auto,
                    node distance = 3cm, 
                    semithick 
                ]
        
                \tikzstyle{every state}=[
                    font = \Huge,
                    draw = black,
                    thick,
                    fill = white,
                    minimum size = 15mm
                ]
                \node[state, label = above:{$(w_{2}, w_{3}, w_{1}, w_{4}, w_{5})$, Men-optimal}] (Mopt) {$\mu^{M}$};
                \node[state, label = above left:{$(w_{4},w_{3},w_{1},w_{2},w_{5})$}] (v1) [below left of=Mopt] {$\mu_{1}$};
                \node[state, label = above right:{$(w_{2},w_{3},w_{5},w_{4},w_{1})$}] (v2) [below right of=Mopt] {$\mu_{2}$};
                \node[state, label = above left:{$(w_{4},w_{2},w_{1},w_{3},w_{5})$}] (v3) [below left of=v1] {$\mu_{3}$};
                \node[state, label = below right:{$(w_{4},w_{3},w_{5},w_{2},w_{1})$}] (v4) [below left of=v2] {$\mu_{4}$};
    
                \node[state, label = below:{$(w_{4},w_{2},w_{5},w_{3},w_{1})$, Women-optimal (Sex-equal)}] (v5) [below  of=v4] {$\mu^{W}$};
    
                \path[->] (Mopt) edge[scale = 2.0] node {$\rho_{1}$} (v1);
                \path[->] (Mopt) edge[scale = 2.0] node {$\rho_{3}$} (v2);
                \path[->] (v1) edge[scale = 2.0] node {$\rho_{2}$} (v3);
                \path[->] (v1) edge[scale = 2.0] node {$\rho_{3}$} (v4);
                \path[->] (v4) edge[scale = 2.0] node {$\rho_{2}$} (v5);
                \path[->] (v2) edge[scale = 2.0] node {$\rho_{1}$} (v4);
				\path[->] (v3) edge[scale = 2.0] node {$\rho_{3}$} (v5);

        \end{tikzpicture}
         \caption{The stable lattice for the profile described in \cref{example:sex-equal} with six stable solutions.
         Matchings are denoted as lists of women matched to men according to their indices, rotations are indicated by $\rho_x$ on arcs.}
         \label{fig:hasse}
\end{figure}
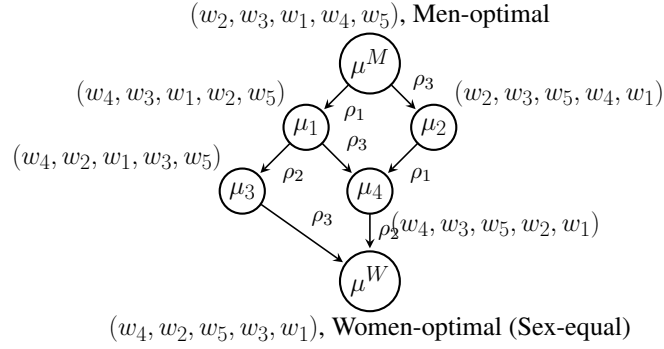

The men-optimal and women-optimal matchings are marked by $\dagger$ and $\ast$ respectively. The stable lattice of this preference profile contains five  matchings as illustrated in \Cref{fig:hasse}. 
A sex-equal stable matching assigns a total ranking of 12 to women and 13 to men, minimizing the welfare difference between the two sides. 
Notice that the sex-equal solution for this profile lies precisely on the extreme point of the stable lattice and is equivalent to the women-optimal solution.\footnote{We refer the readers to Appendix \ref{app:prelim} for details of lattices and rotations.}
\end{example}



Motivated by this observation, we study stable matching problems when both sides of the market have correlated preference lists. We focus on correlated random preference lists sampled from \emph{Mallows distribution models}~\cite{mallows1957non,marden1995analyzing}. The Mallows model is the cornerstone of a variety of ranking problems in machine learning~\citep{vitelli2017probabilistic,awasthi2014learning,lu2014effective} and social choice~\cite{irurozki2019mallows}, and has been shown to correctly capture the preference rankings of individuals in many practical applications~\citep{caragiannis2013noisy,steyvers2009wisdom}. 
Under the Mallows models, preferences are correlated through a reference ranking, where the probability of a preference list to be sampled is inversely proportional to its Kendall-tau distance~\citep{kendall1938new,kendall1948rank} (number of swaps between the two lists) from the reference ranking.
Therefore, focusing on settings with correlated preference lists on both sides of the market we ask the following questions:

\begin{quote}
    \emph{How does the correlation between preference lists of two sides impact the fairness of stable matchings? 
    What algorithmic solutions should we adopt in practice when preferences are correlated according to Mallows models?}
\end{quote}

\subsection{Our Results}
Focusing on markets with correlated preferences on both sides, we empirically investigate the sex-equality of stable matchings through a series of extensive experiments on synthetic data:


\begin{itemize}
    \item \textbf{Stable lattice}: We first focus on the size of the stable lattice when preferences of both sides are drawn from Mallows distributions. 
    We show that the size of the stable lattice is heavily dependent on the relationship between the correlation intensity in preferences of two sides (\Cref{sec:lattice}). In particular, we show that when correlations between the two sides of the market are \textit{symmetric}, that is, preferences are distributed according to the same dispersion parameter, the stable lattice grows rapidly--as it was theoretically proved by \citet{levy2017novel}.
    However, when correlations are \textit{asymmetric}, i.e. distributed according to different dispersion parameters, the size of the stable lattice sharply decreases.
    
    \item \textbf{Asymmetric correlations}: When preferences of the two sides are sampled from different Mallows distributions, we show that in overwhelming majority of cases, a sex-equal stable matching is located at the extreme points of the stable lattice. In \cref{sec:asymmetric} we discuss how this key observation immediately results in a polynomial-time algorithm for finding a sex-equal matching by computing a men-optimal or women-optimal stable matching through the \DA{} algorithm \citep{gale1962college}.
    
    \item \textbf{Symmetric correlations}: When preferences are sampled from the same Mallows distribution, even though the size of the stable lattice may be exponential, the cost of a sex-equal matching is considerably close to the cost of the \DA{} outcome. As the difference between the welfare of men and women is sufficiently small, the importance of which group proposes (men or women) becomes negligible (\cref{sec:symmetric}).
    
    \item \textbf{Comparing the performance of algorithms}:
    In \cref{sec:comparison} we conduct a series of empirical comparisons between well-studied heuristic and procedurally fair algorithms. 
    We show that with correlated preferences, the \DA{} algorithm---with a careful selection of the proposing side---performs as good as the best known local search algorithm with respect to the sex-equality cost even on large instances. This result further justifies the use of the \DA{} algorithm which is more computationally efficient compared to other methods.
\end{itemize}

\subsection{Related Work}

Stability is a key condition for the success and longevity of two-sided markets \cite{roth2002economist}. The seminal work by Gale and Shapley showed that a stable matching always exists and can be found in polynomial time using the Deferred Acceptance algorithm ($\DA{}$) \cite{gale1962college}. The important property of \DA{} is inherent asymmetry: favouring one side at the cost of another \cite{gusfield1989stable}. 
The inherent bias of \DA{} should be taken into account as many central clearinghouses use it as the basis for their matching procedures, including the National Resident Matching Program (NRMP) \cite{roth1999redesign} and the New York City school assignment \cite{abdulkadirouglu2005new}.

Several papers investigate the difference between the expected rank of partners in an optimal and a pessimal stable matching (aka \textit{the welfare gap}) \cite{ashlagi2017unbalanced, ashlagi2020tiered, rheingans2020large, pittel1989average}. In one-to-one balanced markets (i.e. with an equal number of men and women) and when preferences drawn uniformly at random, the expected total rank of agents asymptotically approaches $n\log n$ in their optimal matching and $n^2/\log n$ in their pessimal matching \citep{pittel1989average}. The same order of scores is preserved in markets with constant tier-based preferences, in which agents rank partners proportionally to their real-numbered popularity scores \citep{ashlagi2020tiered}. 

Large real-world markets usually admit an exceedingly small number of stable matchings. For example, for the NRMP hospital-resident matching and for Boston Public School student admission there were only a couple of stable matchings each year \cite{roth1999redesign,pathak2007leveling}. It was shown theoretically, that a lattice becomes essentially a singleton in large one-to-one markets, in which one side draws shorter fixed-sized preference lists from an arbitrary distribution \cite{immorlica2015incentives}. Similar results were demonstrated for many-to-one settings: balanced \cite{kojima2009incentives} and unbalanced \cite{ashlagi2017unbalanced}. At the same time, large stable lattices were theoretically proven in markets with correlated preferences induced by the Watts–Strogatz ``small world'' model (even when unbalanced) \cite{rheingans2020large}, and the Mallows model \cite{levy2017novel}. Empirically, large lattices are found in settings of matching-with-contracts \cite{hassidim2018need}.

\subsubsection*{\textbf{Fairness and stability.}}

The potential exponential size of stable solutions \cite{irving1986complexity} alongside the inherent bias of \DA{}, has motivated the discussion of fairness in stable matching markets. Various algorithms were proposed to ensure fairness towards individuals \cite{freemantwo, feder1992new}, socio-economic groups \cite{abdulkadirouglu2003school, huang2010classified, nguyen2019stable}, regions \cite{kamada2015efficient}, and sides (e.g. hospitals vs. residents) \cite{tziavelis2019equitable, mcdermid2014sex, romero2005equitable}. 
Below, we will briefly discuss a variety of fairness notions that have received attention in the literature of stable matchings. 

The \textit{egalitarian} stable matching, is a stable matching that minimizes the sum of rankings of all partners for men and women, hence maximizes the total welfare of all agents \citep{irving1987efficient}. The \textit{minimum regret} stable matching minimizes the highest rank of agents, therefore minimizing the rank of a partner for the most unsatisfied agent. Egalitarian and minimum regret stable matchings can be found in polynominal time \citep{irving1987efficient, feder1992new}, but do no ensure cross-sided fairness: they both can substantially favor one of the sides, by optimizing the welfare on the level of the whole system (egalitarian), and individual agents (minimum regret). In the \textit{median stable matching}, each agent is matched to its middle favorite partner across all partners from stable matchings; the resulting matching is a median element of a stable lattice. Finding a median stable matching is \NPH{} in general, but it could be done in polynomial time in some restricted families of the rotation poset \citep{cheng2010understanding}. The median stable matching might not satisfy cross-sided fairness if one side prefers their middle stable partners substantially more than the other side. A \textit{procedurally fair} algorithm aims to provide agents an equal probability to affect the resultant stable matching by issuing proposals \cite{tziavelis20}, reducing the set of stable matchings \citep{aldershof1999refined}, or satisfying blocking pairs \citep{roth1990random}. Procedural fairness does not necessarily result in cross-sided fair matchings, and thus, one side may receive much higher welfare compare to the other side \cite{klaus2006procedurally}. 

A \textit{sex-equal stable matching} minimizes the gap between the sum of partners' ranks of men and women. The sex-equal stable matching problem is an \NPH{} problem and fixed parameter tractable with respect to the treewidth of the Hasse diagram of the rotation poset \cite{gupta2017treewidth} and $k$ parameter in the $k$-range model \cite{cheng2021stable}. To tackle the sex-equal stable matching problem various heuristics have been proposed: for example, performing local search series on a stable lattice (iBILS) \citep{viet2020shortlist}, and transforming a stable matching using genetic algorithm \citep{nakamura1995genetic}. While some of these approximate algorithms (e.g. iBILS) perform well in  experiments, they bear no theoretical bounds on the quality of an outcome.

\subsubsection*{\textbf{Correlated and random preferences.}} Several recent works in stable matching have focused on investigating a variety of natural structures of profiles, from profiles generated uniformly at random \cite{ashlagi2017unbalanced, pittel1989average} to random profiles with soft and hard constraints on rankings \cite{knoblauch2009marriage, tziavelis20, ashlagi2020tiered}. The Mallows distribution model drew the attention of researchers as it was shown to realistically capture the preference of agents in several applications involving individual decision-makers \cite{steyvers2009wisdom}. In the Mallows model, preferences are correlated through a reference ranking, representing the objective order of agents' attractiveness. In contrast to the tiered and random utility models, the Mallows model generates a variety of correlated preferences using only a few parameters: two dispersion parameters and two reference rankings; one of each for each side of the market.

\section{Preliminaries}
We start by providing a formal representation of the model and define the necessary properties. 

\paragraph{\textbf{Problem setup.}} An instance of the \emph{stable matching problem} is specified by the tuple $I = \langle M, W, \> \rangle$, where $M$ is a set of $n$ men, $W$ is a set of $n$ women, and $\>$ is a \textit{preference profile} which consists of the preference lists of all men and women. The preference list of any man $m \in M$, denoted by $\>_m$, is a strict total order over all women in $W$ (for any woman $w \in W$, the list $\>_w$ is defined analogously).
The rank of woman $w$ in man $m$'s preference list, $\>_m$, is denoted by $r(w,m)$. For instance, given $\>_m = (w_{2}, w_1, w_3)$, we say $w_2$ is ranked first in $m$'s preference list, i.e., $r(w_2, m) = 1$. Similarly, $r(m,w)$ denotes the rank of $m$ in $\>_w$. 


\paragraph{\textbf{Stable matchings.}} A perfect \emph{matching} is a mapping~$\mu: M \cup W \rightarrow M \cup W$ such that $\mu(m) \in W$ for all $m \in M$, $\mu(w) \in M$ for all $w \in W$, and $\mu(m) = w$ if and only if $\mu(w) = m$. 
Given a matching $\mu$, a man-woman pair $(m, w)$ is called a \textbf{blocking pair}, with respect to the preference profile $\>$, if they prefer each other to their assigned partners under $\mu$, i.e., $w \>_m \mu(m)$ and $m \>_w \mu(w)$.
A matching is \textbf{stable} if it contains no blocking pairs. 

Given a preference profile $\>$, the set of all corresponding stable matchings, $S_{\>}$, forms a \emph{distributive lattice}.\footnote{See \cref{app:prelim} for a detailed explanation of the distributive lattice and an example.} The maximum and minimum points of a stable lattice correspond to the men-optimal ($\mu^{M})$ and the women-optimal ($\mu^{W}$) matchings respectively \citep{gusfield1989stable}, i.e. a matching where all men (respectively women) receive their best stable partners. A men-optimal matching is simultaneously women-pessimal: it matches all woman to their worst stable partners \cite{mcvitie1971stable}.
\Cref{fig:hasse} illustrates a lattice consisting of six stable matchings.
The expected size of a stable lattice (correspondingly the size of $S_\>$) is asymptotic to $e^{-1}\ln n$ when preferences are drawn uniformly at random \citep{pittel1989average} and may grow exponentially with the number of agents \citep{irving1986complexity}.

\paragraph{\textbf{Welfare measure.}}

Let $S_M(\mu)$ denote the sum of ranks of men's partners in matching $\mu$, that is, $S_M(\mu) = \sum_{m \in M} r(\mu(m),m)$. Similarly for women, we let $S_W(\mu)$ be the sum of ranks of women's partners given matching $\mu$, i.e. $S_W(\mu)= \sum_{w \in W} r(\mu(w),w)$. 
Note that the smaller values indicate higher social welfare. 
Thus, a men-optimal matching $\mu^{M}$ is the one that minimizes the sum of rankings for men, that is, $\mu^{M} = \arg\min_{\mu\in S_{\>}} S_M(\mu)$. And similarly, for the women-optimal matching we have $\mu^{W} = \arg\min_{\mu\in S_{\>}} S_W(\mu)$. For simplicity, we will refer to the scores of a men-optimal matching $S_M(\mu^M)$, $S_W(\mu^M)$ as \textbf{$S_M$-optimal} and \textbf{$S_W$-pessimal} scores. Similarly, for a women-optimal matching we write \textbf{$S_M$-pessimal} and \textbf{$S_W$-optimal} to indicate $S_M(\mu^W)$ and $S_W(\mu^W)$.

\paragraph{\textbf{The Deferred Acceptance algorithm.}} Given a preference profile $\>$, the Deferred Acceptance (\DA{}) algorithm, proposed by \citet{gale1962college}, consists of rounds of \emph{proposal} and \emph{rejection} phases and proceeds as follows: 
In each round every man who is currently unmatched proposes to his favorite woman from among those who have not rejected him yet. Each woman tentatively accepts her favorite proposal and rejects the rest. The algorithm terminates in $\mathcal{O}(n^{2})$ when no further proposals can be made.

Given any profile $\>$ as input, the \DA{} algorithm is guaranteed to return a stable matching \cite{gale1962college}. Moreover, the \DA{} algorithm is simultaneously \emph{men-optimal} \cite{gale1962college} and \emph{women-pessimal} \cite{mcvitie1971stable}, i.e. men receive their favorite stable partners among all matchings available to them in $S_{\>}$, and women receive their least favorite stable partners.
The \DA{} algorithm returns the stable matchings at the extreme points of the stable lattice depending on which set (men or women) are proposers. Thus, the \DA{} algorithm with women proposing is respectively women-optimal and men-pessimal.

\subsection{Fair Stable Matchings}

Given a matching $\mu$, the \textit{sex-equality cost} of $\mu$ is the absolute difference between the total welfare of men and women, i.e. $S_M$ and $S_W$ scores. Formally, the sex-equality cost of a matching $\mu$ is defined as
\begin{equation}\label{se}
c(\mu) = | S_M(\mu) - S_W(\mu))|.
\end{equation}

Given preference profile $\>$, a \emph{sex-equal stable matching} is a matching in $S_\>$ that minimizes the sex-equality cost across the stable lattice, that is,
\begin{equation}\label{se_min}
    \mu^* \leftarrow \arg\min_{\mu\in S_{\succ}} (c(\mu)).
\end{equation}

Intuitively, in a sex-equal stable matching the total rank of men's partners is as close as possible to that of women (subject to the stability condition).
\citet{kato1993complexity} showed that for an arbitrary preference profile the problem of finding a sex-equal stable matching is strongly \NPH{}.  However,  it is fixed parameter tractable with respect to the range of the profile -- a metric showing the maximum discrepancy between an agent's worst and best rank in the preference profile \cite{cheng2021stable}. Also, it can be found efficiently when agents have a specific two-dimensional single-peaked model \cite{salonen2018mutually} or if preferences are identical on one side \cite{irving2008stable}.

\subsection{Correlated Preferences}

There are several plausible ways to study preference models generated from uniform distributions~\cite{pittel1989average,ashlagi2017unbalanced} or correlated preferences ~\cite{ashlagi2020tiered,beyhaghi2017effect,tziavelis20, rheingans2020large}. The vast majority of these works focused on profiles wherein the preferences of one side are correlated while the other side is considered uniform. We focus on a more \textit{general} distribution models where \textit{both} sides of the market have correlated preferences through the Mallows model.

\paragraph{\textbf{The Mallows model}} A Mallows distribution is a distance-based probabilistic model for permutations correlated with some common reference~\cite{mallows1957non}. It is parameterized by a reference ranking, $\hat{\pi}$, and a dispersion parameter $\phi\in (0, 1]$.
Let $\pi$ be a permutation of a preference list. For any permutation $\pi$, the Mallows model specifies a probability as follows:
\[
p(\pi|\hat{\pi}, \phi) = \frac{1}{Z} \phi^{\tau(\pi,\hat{\pi})}
\]
where $\tau(\pi, \hat{\pi})$ is a Kendall-tau distance (the number of pairwise inversions) between $\pi$ and $\hat{\pi}$ and $Z$ is a normalization constant with $Z = 1(1+\phi)(1+\phi+\phi^{2})\ldots(1+\phi+\ldots+\phi^{n-1})$.

Note that the dispersion parameter $\phi$ indicates the `intensity' of correlation between the sampled preferences. When $\phi=0$, the correlation is maximal, that is, the distribution mass is entirely on the reference ranking and all preference lists are identical. When $\phi = 1$, the correlation is minimal, the probability mass is distributed evenly between all possible permutations and the Mallows model is equivalent to the \textbf{Uniform distribution model} (also known as \emph{Impartial Culture} \cite{black1958theory} in the computational social choice literature).

For each set involved in the stable matching problem, we consider an independent probabilistic  preference model: one Mallows model for the set of men parameterized by $\hat{\pi}_{m}$ and $\phi_{m}$ and another preference model for women specified by $\hat{\pi}_{w}$ and $\phi_{w}$.
In this way, the preferences of men (similarly women) are globally \textit{correlated} with $\hat{\pi}_{m}$ ($\hat{\pi}_{w}$).\footnote{Throughout the paper, we assume $\hat{\pi}_{m} = \hat{\pi}_{w}$, because one can simply relabel the preferences on one side to create any arbitrary distance.}


\paragraph{\textbf{Symmetric and asymmetric models.}}
We let $\phi_{\Delta} = |\phi_m-\phi_w|$ to represent the \textit{correlation disparity} in preferences of men and women sampled for this particular instance. Simply put, correlation disparity reflects how similar the probabilistic preferences of men are \textit{in comparison} to the similarity of women's preferences. 
We call stable matching instances simulated from the Mallows model with zero correlation disparity, $\phi_{\Delta}=0$,  \textit{symmetric correlation markets}, and \textit{asymmetric correlation markets} when simulated with $\phi_{\Delta}\neq 0$.

\section{Stable Lattice under the Mallows model} \label{sec:lattice}

In this section, we show empirically that in  one-to-one markets with correlated preferences induced by Mallows models the size of the stable lattice depends on the correlation disparity.
The empirical investigations give insights on how searching for a sex-equal matching can computationally vary with the size of a stable lattice in a given market.\footnote{The source code is available at \url{https://github.com/Restel/mallows-smp}}

\subsection{Setup and Preference Sampling }

We generate $1,000$ instances of a stable matching problem for each $n \in [10,150]$ and $\phi_m, \phi_w$ in $[0.1, 1.0]$. When generating a preference profile under probabilistic models, we draw a preference list for each agent \textit{i.i.d} from the Mallows distribution using the sampler from PrefLib library \cite{MaWa13a}. 


In all sampled stable matching instances we built the stable lattice, and measured $S_M$, $S_W$ welfare scores. For each experimental setting, we estimate a target statistic (the median for numerical variables like size and mean for binary variables) over $1,000$ instances along with the confidence intervals. Confidence intervals were obtained using \textit{bootstrap sampling with replacement} with sample size $1,000$ with $100$ repeats. We report the results only for $n = 150$, as for the smaller $n$ the results are qualitatively the same. 
We occasionally considered larger instances with $n=300$, when comparing the size of the lattice of the Uniform and the Mallows ($\phi_m=\phi_w=0.5$) models.\footnote{We use the same instance generation and setup in all remaining sections.}

\subsection{The Size of the Stable Lattice}

\subsubsection*{\textbf{The size of the stable lattice depends on the correlation disparity}}
Our empirical investigation illustrates an intriguing observation about the size of the stable lattice: 
Under the Mallows model the size of the stable lattice is affected by the correlation intensity in men's and women's preferences (measured by the dispersion parameters $\phi_m$, $\phi_w$) as well as the discrepancy between them, i.e. the correlation disparity $\phi_{\Delta}$. 
The size of the lattice is maximum when the preferences of men and women are sampled from a distribution with identical dispersion parameters, $\phi_{\Delta} = 0$, and decreases with the rise of the correlation disparity. This relationship is confirmed by median values, 25\%-75\% percentiles, 95\% confidence intervals, maximum values and holds for all combinations of dispersion parameters (\cref{fig:se_size}). These findings are aligned with theoretical results indicating that when the preferences are sampled from the same distribution (parametrized by identical dispersion parameters, $\phi_m = \phi_w$), the size of the stable lattice grows exponentially in the number of agents \cite{levy2017novel}, making the use of an exhaustive search algorithm impractical in these settings.

More importantly, they reveal that the primary predictor of the lattice size is the correlation disparity between the two sides, and the intensity of the correlations (dispersion parameters) have less significant impact on the lattice size.\footnote{Similar results hold for rotation posets as we discuss in Appendix~\ref{app:rotation}.}

\begin{table}[t]
\centering \small
\begin{tabular}{rllrlr}
  \hline
 & n & Type & Median $\mathcal{L}$ & CI($\mathcal{L}$) & max($\mathcal{L}$)  \\ 
  \hline
1 & 150 & Mallows & 16.00 & (2;256) & 1024  \\ 
  3 & 150 & Uniform & 82.00 & (32;231) & 626 \\ 
  4 & 300 & Mallows & 144.00 & (8;7008) & 147456  \\ 
  5 & 300 & Uniform & 209.00 & (91;567) & 1262  \\ 
   \hline
\end{tabular}
\caption{The size of the stable lattice ($\mathcal{L}$) under the Uniform, and the Mallows distributions with $\phi_m=\phi_w=0.5$ models. $CI(\mathcal{L})$ denotes the 95\% confidence interval.
}
\label{table:lat_size_uni_mal_pol}
\end{table}




\subsubsection*{\textbf{Mallows preferences induce large lattices}}
It is thought that introducing correlation in preferences reduces the stable lattice size compared to the random uncorrelated case \cite{roth1999redesign, ashlagi2017unbalanced, lee2016incentive}. In their empirical work, \citet{roth1999redesign} conjectured that having a common objective ranking for the agents' preferences is one of the main reason behind the observed small lattice sizes. 
Our empirical results suggest that while in general the median size of the lattice was indeed smaller in the Mallows compared to the Uniform distribution, some $\phi$ parameters have much larger maximum values and confidence intervals.\footnote{See \Cref{table:summary_phi} in \cref{app:rotation} for details.} 
For example, in case of asymmetric correlation with $\phi_m=0.5$, $\phi_w = 0.3$, $n=150$, the maximum size of the lattice was larger than that of the Uniform model. Similarly, for symmetric case with $\phi_m=0.5$, $\phi_w = 0.5$, and $n=300$, 95\% confidence interval of a stable lattice size spans from $8$ to $7008$ matchings compared to $91$--$567$ interval under the Uniform distribution (see \cref{table:lat_size_uni_mal_pol}).

\begin{figure}[t]
    \centering
    \includegraphics[width = \columnwidth]{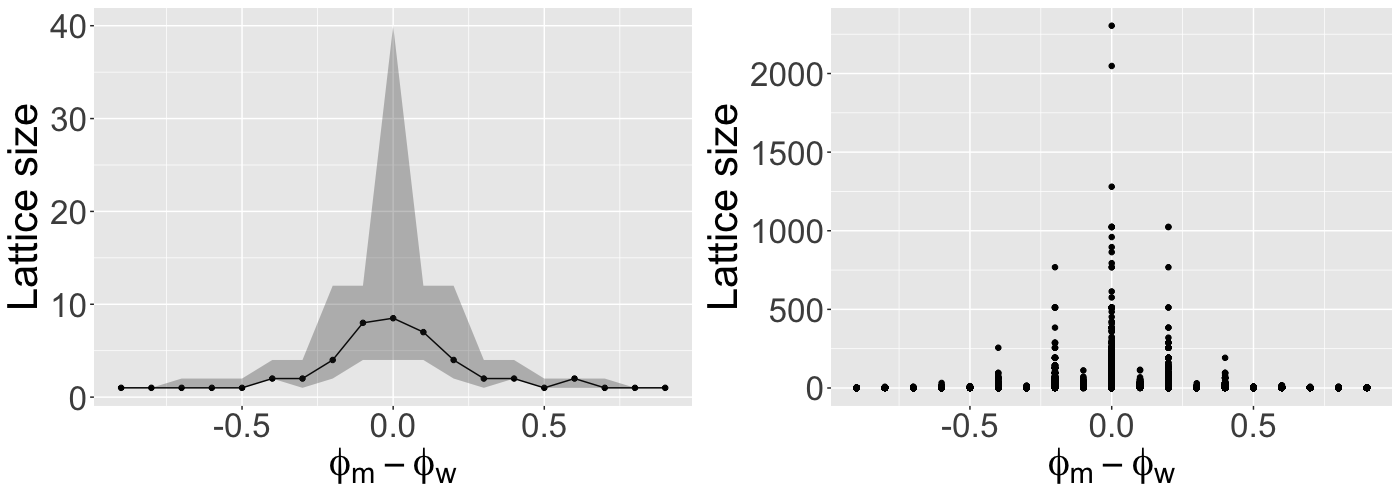}
    \caption{The size of the stable lattice with respect to $\phi_m - \phi_w$ in the Mallows model for $n= 150$: median values (left), maximum values (right). Dotted lines denote the medians, shadowed gray area denotes 1st and 3rd quartiles. Y-axes of left and right figures have different scales.}
    \label{fig:se_size}
\end{figure}

Our results suggest that symmetrically correlated markets have extraordinarily large lattices compared to the Uniform case ($147456$ vs. $1262$), aligning well with the results of \citet{levy2017novel} on asymptotic exponential lattices in the Mallow model. Note, that such large instances appeared for such small number of agents as $300$, while real-world markets often have many more agents \cite{roth1999redesign, pathak2007leveling}. 



In the past, correlation in preferences of agents was considered a sole factor contributing to small lattice sizes \cite{roth1999redesign, ashlagi2017unbalanced}. Our experiments imply that in case of the correlation induced by the Mallows model, this effect depends not on the presence of correlation itself, but rather on its disparity between the two sides i.e. \textit{how strongly the preferences of one set are correlated compared to the preferences of the other set}. Symmetric correlation markets seem to be the worst case scenario, and the lattice tends to grow smaller when the sets are having preferences correlated dissimilarly. This raises the question of whether real-world markets with observed small lattices suffer from such correlation disparity in preference profiles or it should be attributed to other factors. 


\begin{figure*}
\begin{minipage}{0.50\linewidth}
    \centering
    \includegraphics[width = \columnwidth]{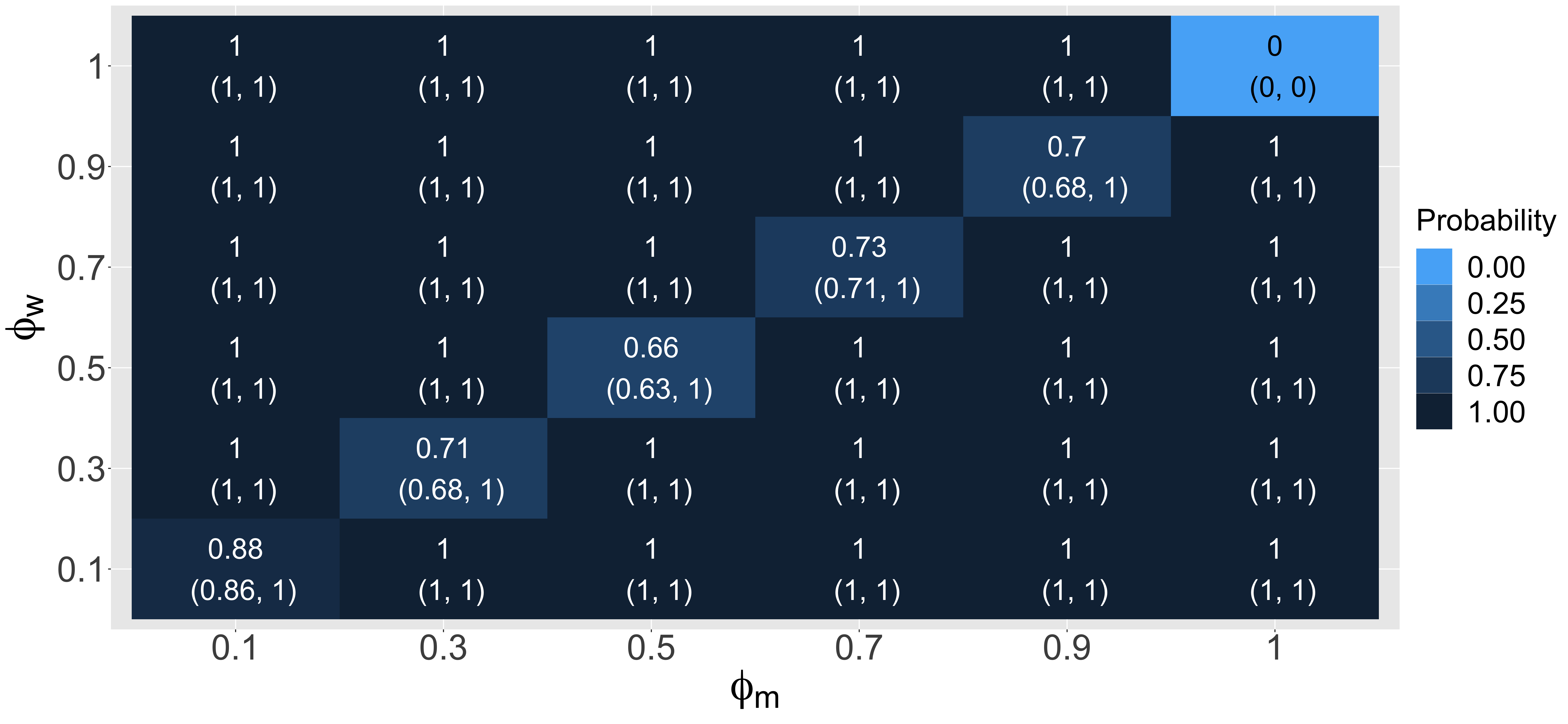}
    \caption{The fraction of instances where a sex-equal solution is an outcome of \DA{} algorithm for  $n = 150$: women proposing if $\phi_m < \phi_w$, men-proposing if $\phi_m > \phi_w$, any side proposing if $\phi_m=\phi_w$. Confidence intervals are given in parentheses.}
    \label{fig:heat_prob}
    \end{minipage}
    \hfill
\begin{minipage}{0.47\linewidth}
    \centering
    \includegraphics[width = \columnwidth]{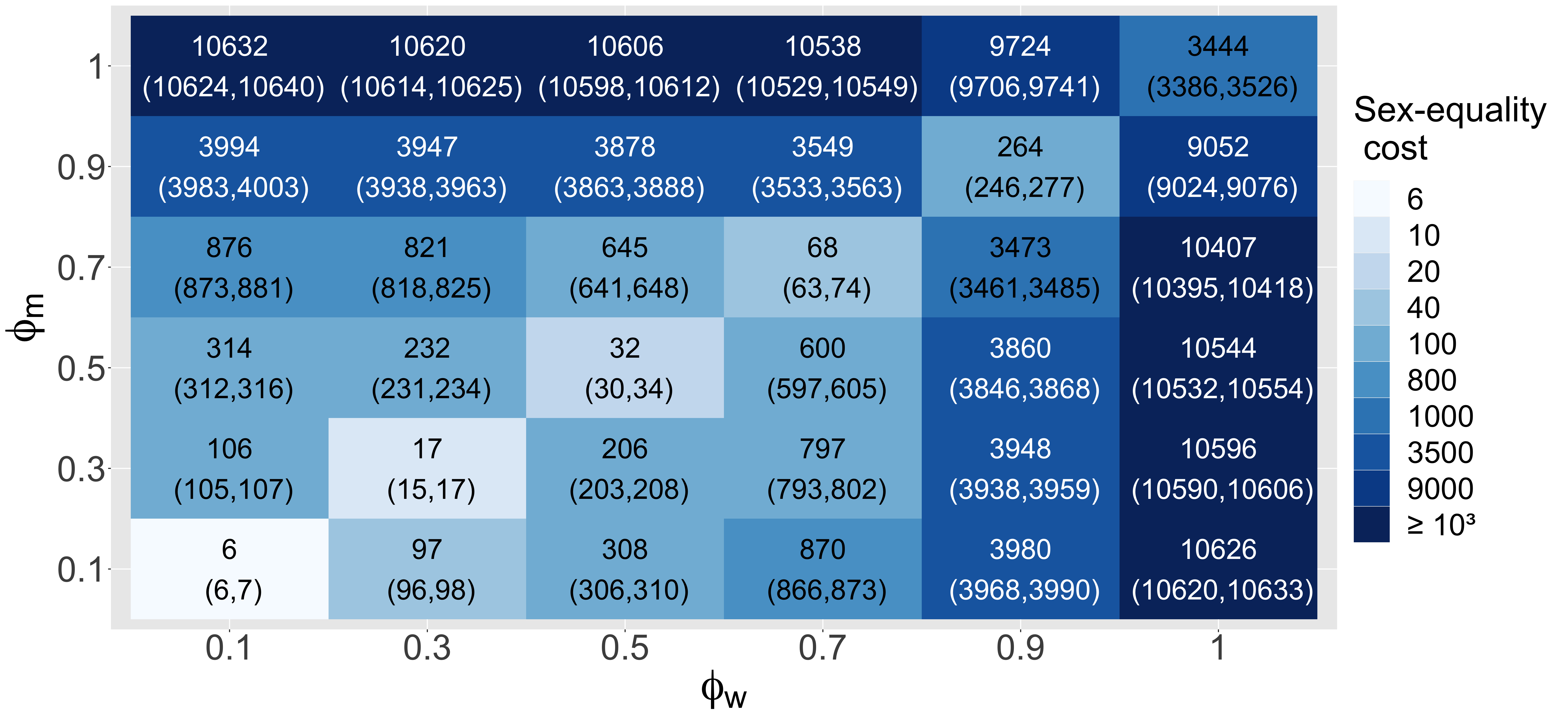}
    \caption{The median sex-equality cost of a \DA{} solution with respect to $\phi_m$, $\phi_w$, for $n=150$. Confidence intervals are given in parentheses.}
    \label{fig:cost_se}
    \end{minipage}
\end{figure*}

\section{Sex-Equal Stable Matchings}\label{sec:sex-equal}

In this section, we discuss how the position of a sex-equal matching within a stable lattice is determined by the relationship between men and women welfare scores, which enables us to find a sex-equal matching using the \DA{} algorithm.

The relationship between $S_M$ and $S_W$ of men- and women-optimal stable matchings dictates the location of a sex-equal solution within the lattice. In particular, it determines whether or not it lies on the extreme points of the lattice. Based on this relationship, we identify two ``easy'' cases, in which \DA{} guarantees to return a sex-equal stable matching, and a ``hard'' case, in which it does not.

\begin{lemma}[\citet{kato1993complexity}]\label{lem:Kato}
A sex-equal stable matching is: \\ 
1) the men-optimal stable matching, $\mu^M$, if $S_M(\mu^M) \geq S_W(\mu^M)$ \\
2) the women-optimal matching, $\mu^W$, if $S_M(\mu^W) \leq S_W(\mu^W)$ \\
3) the stable matching $\mu_Z$ such that $c(\mu_Z) =  \arg\min_{\mu\in S_{\succ}} (c(\mu))$, otherwise.
\end{lemma}

\begin{corollary}\label{cor:sex-equal}
A sex-equal stable matching can be computed in polynomial time when either 
$S_M(\mu^{M}) \geq S_{W}(\mu^{M})$ or $S_M(\mu^{W}) \leq S_{W}(\mu^{W})$.
\end{corollary}

\begin{proof}
Find a men-optimal matching $\mu^M$ by running a men-proposing \DA{} algorithm. Compute $S_M(\mu^M)$ and $S_W(\mu^M)$ and if $S_M(\mu^M) \geq S_W(\mu^M)$, than by \cref{lem:Kato}, Case 1 the men-optimal matching is a sex-equal stable matching. Analogously, find a women-optimal matching $\mu^W$ by running a women-proposing \DA{} algorithm. Compute $S_M(\mu^W)$ and $S_W(\mu^W)$ and if $S_M(\mu^W) \leq S_W(\mu^W)$, than by \cref{lem:Kato}, Case 2 the women-optimal matching is a sex-equal stable solution. The running time of \DA{} and calculating $S_M$, $S_W$ scores is $\mathcal{O}(n^2)$. Thus, a sex-equal stable matching can be found in polynomial time when $S_M(\mu^{M}) \geq S_{W}(\mu^{M})$ or $S_M(\mu^{W}) \leq S_{W}(\mu^{W})$.
\end{proof}

In the first and the second cases of \cref{lem:Kato}, the sign of the difference between the scores $S_M(\mu)-S_W(\mu)$ is preserved across the lattice, in other words $S_M(\mu)$ is either always greater or equal than $S_W(\mu)$ (Case 1), or smaller or equal than $S_W(\mu)$ (Case 2). The third case corresponds to the instances, in which the sign of $S_M(\mu)-S_W(\mu)$ changes across the lattice, and the location of the sex-equal matching within the lattice is not known a priori.

\subsection{Asymmetric Markets ($\phi_m \neq \phi_w$)} \label{sec:asymmetric}

When preferences of men and women have different levels of correlation intensity with the reference rankings, a sex-equal stable matching is often found on an extreme point of the stable lattice: men-optimal in case men have higher correlation and women-optimal otherwise. Empirically, we observe this behavior in all instances, including the largest stable lattices with up to $1024$ matchings (\cref{fig:heat_prob}). We found that agents from the side with a smaller dispersion parameter $\phi$ of the Mallows model, are less satisfied with their optimal matching than their partners with their pessimal one (and the market belongs to Case 1 or Case 2 scenario from \cref{lem:Kato}). This fact might imply that, in asymmetric correlation markets,  ``hard'' cases rarely occur and the \DA{} algorithm can be considered a fair alternative to the exhaustive search, especially since such markets can induce large lattices as shown in \cref{sec:lattice}.


Experimentally, in every stable matching instance belonging to an asymmetric correlation market with $\phi_m < \phi_w$, $S_M$-optimal was larger than $S_W$-pessimal, and therefore a men-optimal matching is sex-equal (\Cref{lem:Kato}, Case 1). Similarly, whenever $\phi_m > \phi_w$, $S_W$-optimal was larger than $S_M$-pessimal, rendering the women-optimal matching sex-equal (\Cref{lem:Kato}, Case 2). This is well illustrated by the distribution of $S_M$ and $S_W$ scores of men- and women-optimal matchings for various combinations of $\phi_m$, $\phi_w$.  \Cref{fig:M_W_0.9} illustrates this relationship when men sample their preferences with $\phi_m = 0.9$, and women have their dispersion parameter in range $[0.1, 1.0]$ (the graphs for other values of $\phi_m$ are analogous). In the part of the graph, where women's preferences have a lower dispersion parameter $\phi_w<\phi_m$, the distribution of $S_W$-scores of a women-optimal matching lies strictly above the distribution of its $S_M$-scores with no overlap.\footnote{A more fine-grained analysis with $\phi_w \in \{0.9, 0.91, 0.92, \ldots, 0.99\}$ confirms the same finding (see \cref{app:fine-grained}).}

\begin{figure}[h!]
    \centering
     \includegraphics[width= .6\columnwidth]{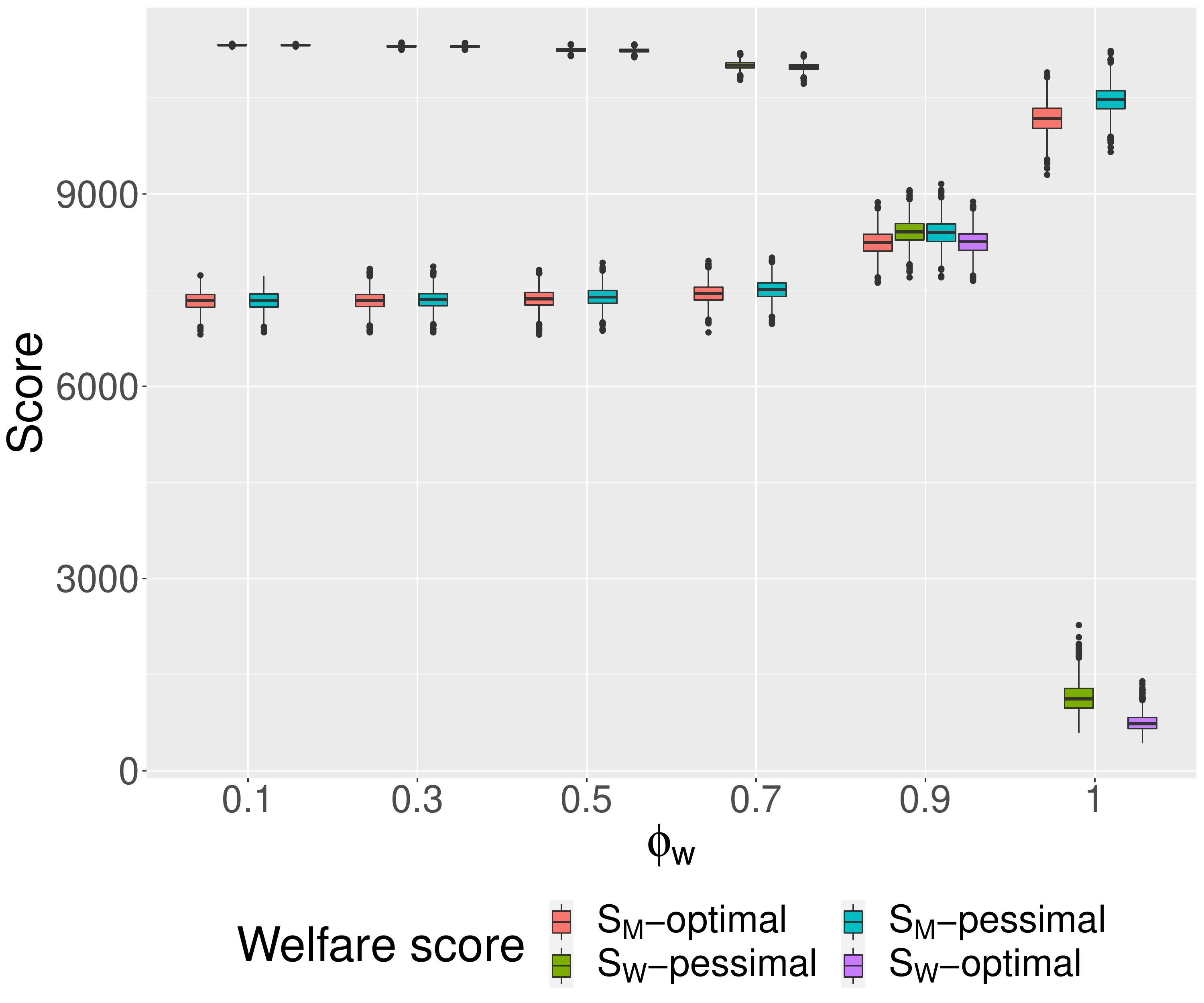}
    \caption{Illustrating $S_M$, $S_W$ scores under the Mallows model with $\phi_m = 0.9$ and varying $\phi_w$ based on 1000 samples.
    }
    \label{fig:M_W_0.9}
\end{figure}

\subsubsection*{\textbf{Implications on the \DA{} algorithm.}}
Interestingly, increasing the variability of women's preferences, results in 1) higher welfare for women in optimal and pessimal matchings and 2) a higher number of proposals from men. The total number of proposals in the \DA{} algorithm corresponds to the optimal score of the proposing side \cite{pittel1989average}. In \cref{fig:M_W_0.9} we can see a noticeable increase in the $S_M$-scores with the increase of $\phi_w$ as well as a decrease in $S_W$ scores (same pattern occurs for other values of $\phi_m$). Imagine that the proposing side samples its preferences from the model with a smaller dispersion parameter, and hence has less variable preferences and demonstrates a higher competition for partners (than the opposite side). Intuitively, during the course of \DA{} those agents would propose to the same set of partners over and over again and their potential partners get to choose from a diverse set of proposals, which improves the rank of their matches. This makes the total rank of proposers' partners greater than that of the agents from the accepting side. This fact distinguishes the Mallows model from the Uniform model in which the proposing side gets a substantially smaller score than the accepting side \citep{pittel1989average}. 


\subsection{Symmetric Markets ($\phi_m = \phi_w$)} \label{sec:symmetric}

 When the preferences of men and women have the same intensity of correlation, reflected in zero correlation disparity, the lattice size can quickly become very large. The size of the stable lattice is expected to be asymptotically exponential  \citep{levy2017novel}, which is confirmed in our experiments: even for small number of agents as $n=300$, there were significantly large (compared to the Uniform model) lattices with up to $147456$ matchings (\cref{table:lat_size_uni_mal_pol}). Despite the large size of the stable lattice, the cost of a sex-equal solution is considerably close to the cost of \DA{} outcomes (\Cref{fig:cost_se}). In this case, the \DA{} algorithm can be considered cost-effective as the exhaustive search takes significantly more time for large lattices and gives only a moderate gain in the sex-equality cost. 

We discussed in \cref{sec:asymmetric}, that under the Mallows model, $S_M$-optimal and $S_W$-pessimal scores (similarly $S_W$-optimal, $S_M$-pessimal) are close to each other and their distributions overlap significantly. This can explain the small sex-equality cost of \DA{} outcomes (see \Cref{fig:M_W_0.9}). It could indicate, that under this scenario, the expected score of the proposing side under \DA{} is similar to the expected score of the receivers. This finding differs sharply from the Uniform model wherein the expected pessimal score is substantially greater than the optimal one \cite{pittel1989average, ashlagi2020tiered}. As the difference between the welfare of men and women is sufficiently small, the importance of which group proposes (men or women) in \DA{} becomes negligible.
Hence, not only can the \DA{} algorithm be adopted to ensure sex-equality, but it may also  provide a natural framework for achieving \textit{procedural fairness} as it has recently been studied \cite{tziavelis20,tziavelis2019equitable}.

\subsection{\textbf{Implications on manipulation strategies}}
In our experiments, the Mallows model with symmetric and asymmetric correlations results in a considerably small difference between the pessimal and the optimal scores for one side (\textit{the welfare gap}), for example, between $S_W$ pessimal and $S_W$-optimal (\Cref{fig:M_W_0.9}). This difference represents how much on average an agent can improve by shifting from her worst stable partner to her best one, and therefore has implications on manipulation incentives \cite{ashlagi2017unbalanced, rheingans2020large}. Since the welfare gap is small for both men and women, in symmetric and asymmetric correlation settings, the accepting side has a limited scope/incentive for manipulations in practice. This observation suggests that one-to-one markets with Mallows preferences might be less susceptible to manipulation compared to the Uniform model, for which a larger welfare gap, and subsequently, higher incentive for manipulation has been established \cite{pittel1989average} and shown empirically \cite{teo2001gale, HUV21accomplice}.

\begin{figure*}
\begin{minipage}{0.48\linewidth}  
    \centering
    \includegraphics[width = \columnwidth]{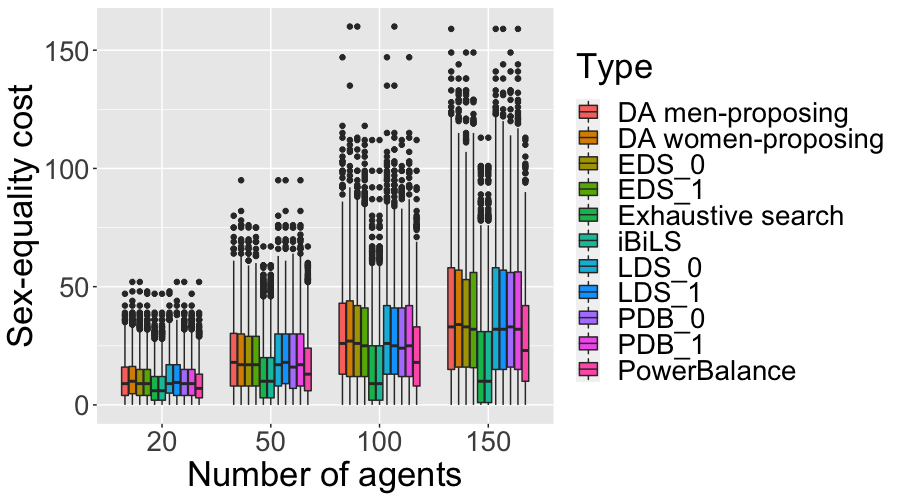}
    \captionof{figure}{Mallows model (symmetric): comparing the sex-equality cost of various matching algorithms when $\phi_m=\phi_w=0.5$. Subscript 0 indicates that men start proposing first, and 1 indicates women propose first.}
    \label{fig:alg_perf_M}
     \end{minipage}
     \hfill
\begin{minipage}{0.48\linewidth}
    \centering
    \includegraphics[width = \columnwidth]{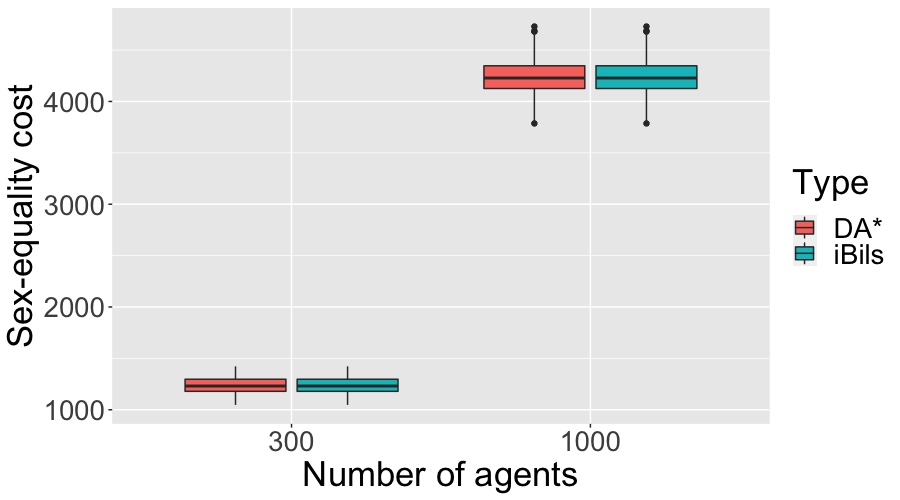}
    \captionof{figure}{Mallows model (asymmetric): the sex-equality cost of iBILS and a pre-processed Deferred Acceptance (\DA{}*) for an asymmetric correlation market with $\phi_m = 0.5$, $\phi_w=0.7$.}
    \label{fig:IBILS_DA_comparison}
    \end{minipage}
\end{figure*}

\section{Comparing with Other Algorithms} \label{sec:comparison}
We compare the \DA{} algorithm with several state-of-the art algorithms (as implemented in \cite{tziavelis20}) such as multiple variants of the procedurally fair algorithms, and the best known heuristic algorithm when preferences are drawn from the Mallows model.\footnote{The comparison of algorithms for the Uniform model can be found in \cref{app:alg_uni}.} We use an \textit{exhaustive search} algorithm as our baseline to find sex-equal stable matchings \cite{tziavelis2019equitable}.
%
In particular, we test several procedures allowing both sides to issue and receive proposals such as \textit{Late Discontent Dispension (LDS)}, \textit{Early Discontent Dispension}, and \textit{Powerbalance} \cite{tziavelis20}. We also measure the performance of an enhanced \textit{bidirectional local search} \textit{iBILS}, which is one of the best-performing existing heuristic for finding sex-equal stable matchings.

\subsection{Experimental Setup}

We evaluate the sex-equality cost and the run-time of these algorithms for the number of agents varying from 20 to 1000. For the symmetric correlation case, we select $\phi_m=\phi_w=0.5$ as these parameters results in the largest lattice size among all combinations of $\phi_m$ and $\phi_w$ (see \cref{table:summary_phi}). We sample 1000 instances of the stable matching for each experimental setting.

The \textit{exhaustive search} algorithm uses the \textit{break-marriage} operation of \citet{mcvitie1971stable}: breaking a single man's assignment initiates a chain of proposals and rejections. It terminates when either all women reject the last man, or his proposal is accepted by a single woman. The latter case results in a new stable matching. All the elements of the stable lattice can be found by applying this operation recursively from the men-optimal solution  \citep{mcvitie1971stable}.\footnote{An example of the break-marriage operation is provided in \cref{app:prelim}.}

The \textit{enhanced bidirectional local search}, iBILS, uses rotations to generate the successors of a stable matching to search over the stable lattice starting from its extreme points \cite{tziavelis2019equitable, viet2016bidirectional}.

\subsection{Fairness and Running Time}

\textbf{Symmetric correlation.} For symmetric correlation markets we used a smaller number of agents (up to 150), to avoid instances with extremely large lattices, for which exhaustive search would not terminate in a reasonable time (\cref{sec:symmetric}). Under symmetric correlations, stable lattice can grow exponentially  \citep{levy2017novel}.
In our experiments, all algorithms performed similarly on average with respect to the sex-equality cost (Figure \ref{fig:alg_perf_M}). Among them, the iBILS algorithm performed better, and in most instances, found the sex-equal stable matching that was computed by the exhaustive search algorithm.
These results and the good performance of other algorithms are justified by the fact that under the Mallows model, $S_M$ and $S_W$ scores of all stable matchings are close to one another (i.e. the welfare gap between the two sets is small). Therefore, many stable matchings have costs close to that of the the sex-equal solution.  

While the \DA{} and procedurally fair algorithms  give higher sex-equality costs (on average twice as much as the exhaustive search and the iBILS algorithms), they have much smaller running times. For larger instances, this discrepancy between the running times will become even more pronounced. The exhaustive search has exponential complexity in the worst case. The running time of the iBILS algorithm depends on the stable lattice size and its configuration: it is $\mathcal{O}(dkn^2)$, where $d$ is the maximum search depth and $k$ is the maximum width of the lattice \cite{viet2020shortlist,tziavelis2019equitable}. The \DA{} algorithm and procedurally fair algorithms run in polynomial-time.

\noindent\textbf{Asymmetric correlation.}
In asymmetric correlation markets, we focused on comparing the performance of iBILS and \DA{} with additional pre-processing step (aka \DA{}*). Given the findings in \cref{sec:asymmetric}, we modify the \DA{} algorithm by deciding which side should propose according to the dispersion parameters: the side with smaller dispersion parameter proposes, i.e. when $\phi_m<\phi_w$ men propose, otherwise women propose. In real-world markets, dispersion parameters can be inferred from preference profiles in polynomial time when the reference ranking is known \cite{irurozki2014sampling}.



\cref{fig:IBILS_DA_comparison} presents the sex-equality cost of \DA{}* compared to the iBILS algorithm for large markets with 300 and 1000 agents. In almost all instances, \DA{}* finds the same stable matching as iBILS. More importantly, as shown in \Cref{fig:IBILS_DA_time_comp}, the run-time of the \DA{}* remains in $\mathcal{O}(n^2)$ while iBILS requires more steps to estimate a sex-equal solution.
Given the performance of iBILS in most cases, our results suggest that \DA{}* is a reasonable cost-effective fair algorithm for matching markets with asymmetric correlations.

\begin{figure}[h!]
    \centering
     \includegraphics[width= .6\columnwidth]{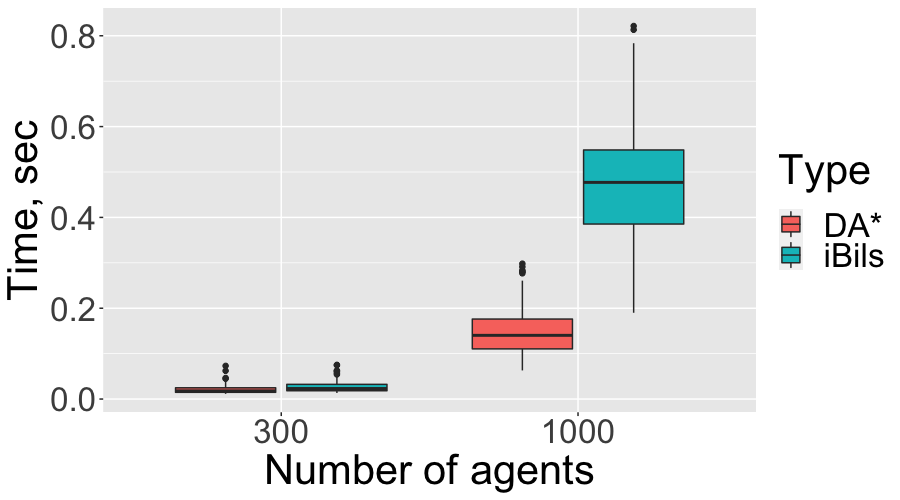}
    \caption{Mallows model  (asymmetric): the running time of iBILS and a pre-processed Deferred Acceptance (\DA{}*) for an asymmetric correlation market with $\phi_m = 0.5$, $\phi_w=0.7$ (other combinations of parameters give analogous results).
    }
    \label{fig:IBILS_DA_time_comp}
\end{figure}


\section{Concluding Remarks}

For decades, market design has played a crucial role in setting the foundations of decision-making through mechanism design. 
While the theoretical guarantees sometimes fall short due to incompatibilities between axiomatic requirements or challenges imposed by computational complexity, empirical evaluations can still shed light on practical implications and guide the design of new mechanisms. This in turn could help produce new problems and directions for traditional theory.

In this spirit, our extensive empirical investigations revealed an intriguing relation between the disparity of preferences of both sides and the fairness of stable matchings. We showed that surprisingly, the primary factor affecting the sex-equality cost of the \DA{} algorithm is the difference between dispersion parameters of both sides. These observations suggest that the \DA{} algorithm could still be a prime candidate both in symmetric and asymmetric markets, further justifying its wide use in practice.

From the theoretical perspective, investigating the bounds of fairness in stable matchings under correlated preferences is certainly an interesting future direction.
From the practical perspective, a noteworthy, and perhaps more important, direction is to investigate the structure of preferences in various matching settings (e.g. refugee matching, residency matching, and school choice) and develop domain-specific models that can correctly capture the correlations between the preferences with high accuracy. We believe that these models, alongside with theoretical studies, should shape governmental and societal policies---prescribed by social planners---in the adoption of suitable mechanisms in each specific domain.

\section*{Acknowledgments}
Hadi Hosseini acknowledges support from NSF IIS grants \#2052488 and \#2107173.
We are grateful to the anonymous referees for their valuable comments and suggestions.




\bibliographystyle{plainnat} 
\bibliography{ref,references}

\clearpage

\begin{center}\huge
    \textbf{Supplementary Material}
\end{center}
\appendix

\section{Additional Preliminaries} \label{app:prelim}

In this section, we give additional preliminaries, definitions, and examples to provide a complete view to the readers.

\subsubsection*{\textbf{Dominance}}
Stable matching $\mu'$ dominates stable matching $\mu$ if it assigns all men a partner at least as good as in $\mu$ and a better partner to at least one man \citep{gusfield1989stable}. Formally, $\mu'$ dominates $\mu$ if $\forall m \in M$ $r(\mu'(m), m) \leq r(\mu(m),m)$ and $\exists m^{*} \in M$ such that $r(\mu'(m^{*}), m^{*}) < r(\mu(m^{*}), m^{*})$.

\subsubsection*{\textbf{Stable lattice}} 

Given a preference profile $\>$, the set of all corresponding stable matchings, $S_{\>}$, forms a \emph{distributive lattice} under a dominance relation (referred to as a \textit{stable lattice} hereinafter). Any pair of stable matchings in a stable lattice has a unique least upper bound (also called a \textit{join}) and a unique greatest lower bound (\textit{a meet}) within the lattice. In a join of two stable matchings, $\mu_1$, $\mu_2$, every man from $M$ is given its best partner among $\mu_1$, $\mu_2$. Similarly, in a meet all men are given their worse partners from $\mu_1$, $\mu_2$. 

In a stable lattice supremum, the men-optimal matching $\mu^{M}$, all men are matched to their favorite partners among all potential stable partners. It dominates every other matching in $S_\>$, formally $\forall m\in M$, $r(\mu^{M}(m), m) \leq r(\mu'(m), m)$ for all $\mu'\in S_{\>}$. Similarly, in a stable lattice infimum, the men-pessimal matching $\mu^{W}$, all men are given their worse stable partners. Men-pessimal matching is dominated by every other stable matching. Interestingly, a stable lattice infimum is simultaneously women-optimal matching: each woman is matched to her best stable partner in $S_\succ$, i.e.,  $\forall w\in W$, $r(\mu^{W}(w), w) \leq r(\mu'(w), w)$ for all $\mu'\in S_{\>}$. \Cref{fig:hasse} illustrates a lattice consisting of six stable matchings.

\subsubsection*{\textbf{Rotation poset.}}

For any stable lattice $S_\succ$ there exists an associated partially ordered set (\textit{a rotation poset}, $\mathcal{R}(S_\succ)$), whose downsets are in one-to-one correspondence with the matchings of the lattice \citep{gusfield1989stable}. The elements of the rotation poset---aka rotations---are specific sequences of pairs matched in some stable matching. For each man $m$ in a rotation the following holds: 1) the woman in his pair is his current partner $\mu(m)$; 2) the woman to the right of $m$ (or the first woman if $m$ is the last man), $w_{next}(m)$, would be the first to accept his proposal if he is disallowed to propose to $\mu(m)$ within the course of \DA{}. Formally, a rotation $\rho = (m_1,w_1), (m_2, w_2), ..., (m_k, w_k)$ is a sequence of pairs in a matching such that $m_i$ is matched to $w_i$, and $w_{(i+1)\ mod\ k}$ is the first woman in $m_i$'s preference list who prefers him to her current partner. For specific $\mu$ and $\rho$, if for every man $m_i \in \rho$ the above conditions holds, the rotation $\rho$ is said to be \textit{exposed} in $\mu$. One rotation can be exposed in multiple stable matchings of the same problem instance; a stable matching can have multiple exposed rotations \cite{gusfield1989stable} .

Given $\rho$ exposed in $\mu$, if every man $m_i$ from the rotation is matched to the woman to the right of him, $w_{i+1}$ (where $(i+1)$ is taken modulo size of $\rho$), and every man not in the rotation is assigned his current partner, then the resulting matching is also stable \citep{irving1986complexity}. Rotations can be thought of as the transition operations between the adjacent elements of the lattice. The rotation poset is a useful structure for studying the stable matching problem \citep{bhatnagar2008sampling, gusfield1987three, cheng2010understanding} as it contains $\mathcal{O}(n^{2})$ elements (in contrast to the exponential size of the stable lattice in the worst case) and yet encodes all stable matchings. We refer the reader to \citep{gusfield1989stable} for a more detailed discussion on the subject of rotations. 

\subsubsection*{\textbf{Example from \Cref{fig:hasse}: Stable Lattice and Rotations}}
In this section, we illustrate the concept of stable lattice, rotation poset and break-marriage operation using the preference profile given in \Cref{example:sex-equal}. 

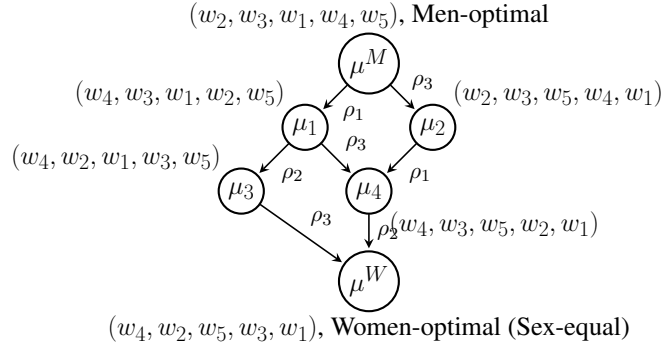
\begin{figure}[t]
    \centering
        \centering
        \begin{tikzpicture}[
                    every edge quotes/.append style={font=\LARGE},
                    every label/.append style={font = \Huge},
                    level 1/.style={sibling distance=20mm},
                    scale=0.4, 
                    every node/.append style={transform shape},
                    > = stealth, 
                    shorten > = 1pt, 
                    auto,
                    node distance = 3cm, 
                    semithick 
                ]
        
                \tikzstyle{every state}=[
                    font = \Huge,
                    draw = black,
                    thick,
                    fill = white,
                    minimum size = 15mm
                ]
                \node[state, label = above:{$(w_{2}, w_{3}, w_{1}, w_{4}, w_{5})$, Men-optimal}] (Mopt) {$\mu^{M}$};
                \node[state, label = above left:{$(w_{4},w_{3},w_{1},w_{2},w_{5})$}] (v1) [below left of=Mopt] {$\mu_{1}$};
                \node[state, label = above right:{$(w_{2},w_{3},w_{5},w_{4},w_{1})$}] (v2) [below right of=Mopt] {$\mu_{2}$};
                \node[state, label = above left:{$(w_{4},w_{2},w_{1},w_{3},w_{5})$}] (v3) [below left of=v1] {$\mu_{3}$};
                \node[state, label = below right:{$(w_{4},w_{3},w_{5},w_{2},w_{1})$}] (v4) [below left of=v2] {$\mu_{4}$};
    
                \node[state, label = below:{$(w_{4},w_{2},w_{5},w_{3},w_{1})$, Women-optimal (Sex-equal)}] (v5) [below  of=v4] {$\mu^{W}$};
    
                \path[->] (Mopt) edge[scale = 2.0] node {$\rho_{1}$} (v1);
                \path[->] (Mopt) edge[scale = 2.0] node {$\rho_{3}$} (v2);
                \path[->] (v1) edge[scale = 2.0] node {$\rho_{2}$} (v3);
                \path[->] (v1) edge[scale = 2.0] node {$\rho_{3}$} (v4);
                \path[->] (v4) edge[scale = 2.0] node {$\rho_{2}$} (v5);
                \path[->] (v2) edge[scale = 2.0] node {$\rho_{1}$} (v4);
				\path[->] (v3) edge[scale = 2.0] node {$\rho_{3}$} (v5);

        \end{tikzpicture}
         \caption{The stable lattice for the profile described in \cref{example:sex-equal} with six stable solutions.
         Matchings are denoted as lists of women matched to men according to their indices, rotations are indicated by $\rho_x$ on arcs.
         }
         \label{fig:hasse1}
\end{figure}

We run a men-proposing \DA{} to obtain the men-optimal matching, $\mu^M$ = ($w_2$, $w_3$, $w_1$, $w_4$, $w_5$), where women are matched to men according to their position in the matching, e.g. $w_2$ is matched to $m_1$, $w_3$ to $m_2$ and so on.

To identify a rotation exposed in $\mu^M$ we will use the concept of \textit{shortlists} \citep{irving1987efficient}. Given a stable matching $\mu$, let each woman $w \in W$ remove from her list all men, that she does not prefer over her current partner, i.e. every $m'$ s.t. $\mu(m) \succ_w m'$. Removed men delete woman $w$ from their preferences as well. In the resulting male-oriented shortlist, the following holds for every $m \in M$ : 1) $m$'s current partner is the first entry in his list $\succ_m$; 2)  the first woman who prefers $m$ to her partner, $w_{next}(m)$, is the second entry in $\succ_m$. Also, for every woman $w \in W$ her current partner is the last entry in her preference list $\succ_w$ \cite{irving1987efficient}.
\begin{table}[t]
\label{shortlist}
\small
   \centering
  \begin{tabularx}{\linewidth}{XXXXXXXXXXXXXXX}
           $m_1\colon$& $w_2^\dagger$& $w_4$& $w_5$& $w_1$& &&
           $w_1\colon$&$m_4$& $m_2$& $m_1$& $m_5$& $m_3^\dagger$\\
            $m_2\colon$& $w_3^\dagger$& $w_2$& $w_4$& $w_1$& $w_5$ &&
            $w_2\colon$&$m_2$& $m_4$& $m_1^\dagger$& & \\
            $m_3\colon$&$w_1^\dagger$& $w_5$& & & &&
            $w_3\colon$&$m_4$& $m_2^\dagger$& & & \\
            $m_4\colon$&$w_4^\dagger$& $w_2$& $w_3$& $w_1$& $w_5$ &&
            $w_4\colon$&$m_2$& $m_1$& $m_4^\dagger$& & \\
             $m_5\colon$&& & $w_5^\dagger$& $w_1$&  &&
             $w_5\colon$&$m_1$&$m_4$& $m_2$& $m_3$& $m_5^\dagger$
 \end{tabularx}
 \caption{The shortlist of men-optimal matching of \Cref{example:sex-equal} instance}
\end{table}

To identify an exposed rotation, we will take the first man $m_1$ and put it in front of a candidate rotation along with his current partner $\rho_1 = (m_1, \mu(m_1))$. The next pair in the sequence must contain $w_{next}(m_1)$ and her partner $\mu(w_{next}(m_1))$, i.e. $\rho_1 = (m_1, \mu(m_1)), (\mu(w_{next}(m_1)), w_{next}(m_1))$. We will repeat building the candidate rotation in this fashion. The rotation becomes valid as soon as one of the men appears twice in the same candidate rotation. After this happens, we will discard all pairs outside of the cycle: the remaining ordered sequence of pairs constitutes a rotation exposed in $\mu$. In order to identify all rotations exposed in the same matching, the same procedure must be repeated for all men that did not appear in a valid rotation and were not discarded in the course of building one.

We will identify all rotations exposed in $\mu^{M}$ to obtain several stable matchings adjacent to it. We start a candidate rotation with putting $m_1$ and his partner $w_2$ to the beginning: $\rho_1 = (m_1, w_2)$. The first woman who prefers $m_1$ over her current partner is $w_4$, and we append her and her partner as the next pair in the sequence: $\rho_1 = (m_1, w_2)(m_4, w_4)$. The next pair in the sequence will be ($m_1$, $w_2$) as $w_{next}(m_4) = w_2$. As it is the second appearance of $m_1$ in the sequence we found a rotation exposed in $\mu^M$. We continue in this fashion with the next man $m_3$ as a founder of a new candidate rotation, that leads to $\rho_3 = (m_3, w_1), (m_5, w_5)$. Using $m_2$ as the first man of a rotation leads to already found $\rho_1$. 

When rotation $\rho_1$ is applied to $\mu^M$, $m_1$ is matched to $w_4$, $m_4$ to $w_2$, giving stable matching $\mu_1 = (w_4, w_3, w_1, w_2, w_5)$. Similarly, the application of $\rho_3$ gives $\mu_2=(w_2, w_3, w_5, w_4, w_1)$. By finding and applying rotations exposed in $\mu_1$, $\mu_2$ one can recreate the lattice from 
\cref{example:sex-equal} with 6 elements.

One can also use a \textbf{break-marriage algorithm} to find all the elements of the lattice: it is a recursive procedure that starts with breaking assignments of the fist man in the men-optimal matching, and letting him propose to the next woman who has not rejected him yet. This initial change in the matching initiates a chain of proposals and rejections between the agents and terminates when either all women reject some man (the last one to propose), or his proposal is accepted by a single woman. The latter case results in a new stable matching. By applying this operation recursively to all new stable matchings one can find all the elements of the stable lattice \citep{mcvitie1971stable}

Let apply break-marriage to the fist man in the men-optimal matching $\mu^M$: $m_1$ is assigned free and propose to his next favorite woman in the list $w_4$. She prefers his proposal to her current partner and is matched to $m_1$ tentatively. Man $m_4$ proposes to $w_2$, who accepts the proposal since she is currently single. The new matching is $\mu_1$ and the current breakmarriage round will now assign $m_2$ in $\mu^M$ free and tries to identify a new stable matching. If breakmarriage is applied to all men in the matching, it will result in repetitive stable matchings, so a more efficient modification of the algorithm exists \cite{mcvitie1971stable}.

\section{Summary Statistics for Stable Lattice and Rotation Poset}\label{app:rotation}

\begin{figure}[t]
    \centering
    \includegraphics[width = \columnwidth]{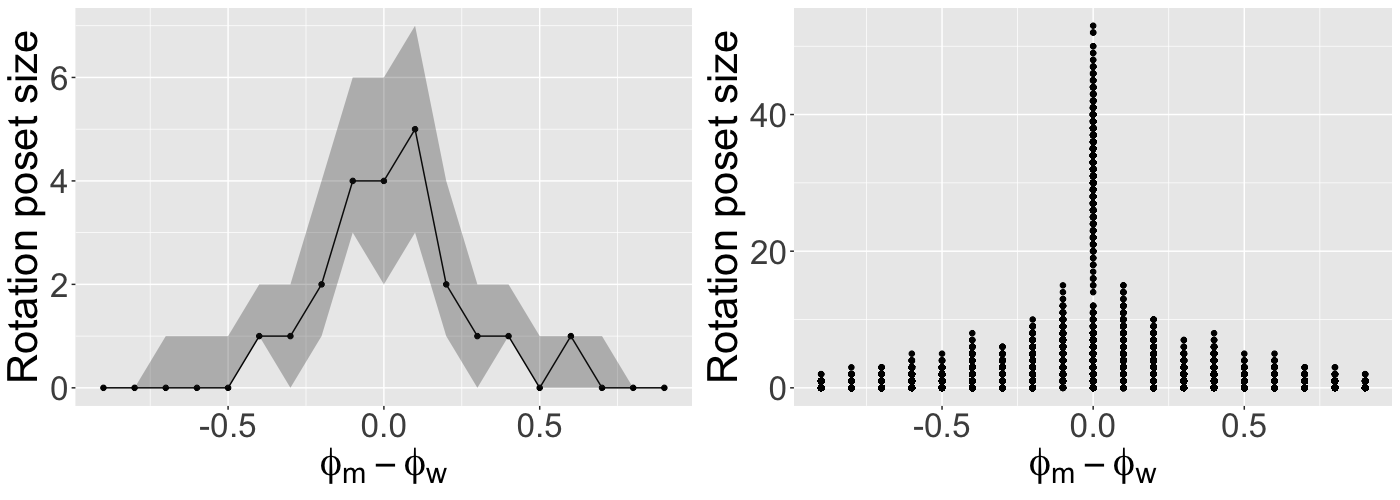}
    \caption{The size of the rotation poset with respect to $\phi_m - \phi_w$ in the Mallows model for $n= 150$: median values (left), maximum values (right). Dotted lines denote the medians, shadowed gray area denotes 1st and 3rd quartiles. Y-axes of left and right figures have different scales}
    \label{fig:rot_size}
\end{figure}

\begin{table}[t]
\centering 
\small
\begin{tabular}{rllrrl}
  \hline
 & n & Type & max($\mathcal{R}$) & Median $\mathcal{R}$ & CI($\mathcal{R}$) \\ 
  \hline
1 & 150 & Mallows &  11 & 4.00 & (1;8) \\ 
  3 & 150 & Uniform & 53 & 33.00 & (19;45) \\ 
  4 & 300 & Mallows & 18 & 8.00 & (3;13) \\ 
  5 & 300 & Uniform & 85 & 60.00 & (44;76) \\ 
   \hline
\end{tabular}
\caption{The size of the rotation poset ($\mathcal{R}$) under the Uniform, and the Mallows with $\phi_m=\phi_w=0.5$ models.
}
\label{table:rot_pos}
\end{table}

\begin{table*}[!htbp]
\centering
\begin{tabular}{rlrrrcrrrr}
   \hline
 & $\phi_{w}$ & Q2($\mathcal{L}$) & Q1($\mathcal{L}$) & Q3($\mathcal{L}$) & max($\mathcal{L}$) & max($\mathcal{R}$) & Q2($\mathcal{R}$) & Q1($\mathcal{R}$) & Q3($\mathcal{R}$) \\ 
  \hline
    & ~ & ~ & ~ & ~ & \bf{$\phi_m = 0.1$} &  ~  & ~ & ~ & ~  \\
  \hline
\textbf{1} &\textbf{0.1 }& \textbf{2.00 }& \textbf{1.00 }&\textbf{ 4.00} & \textbf{128} &   \textbf{7} & \textbf{1.00} & \textbf{0.00} &\textbf{ 2.00} \\ 
  2 & 0.3 & 4.00 & 2.00 & 8.00 & 128 &   7 & 2.00 & 1.00 & 3.00 \\ 
  3 & 0.5 & 2.00 & 2.00 & 4.00 &  32 &   5 & 1.00 & 1.00 & 2.00 \\ 
  4 & 0.7 & 2.00 & 1.00 & 2.00 &  32 &   5 & 1.00 & 0.00 & 1.00 \\ 
  5 & 0.9 & 1.00 & 1.00 & 1.00 &   8 &   3 & 0.00 & 0.00 & 0.00 \\ 
  6 & 1 & 1.00 & 1.00 & 1.00 &   4 &   2 & 0.00 & 0.00 & 0.00 \\ 
  \\ 
    \hline
    & ~ & ~ & ~ & ~ & \bf{$\phi_m = 0.3$} &  ~  & ~ & ~ & ~  \\
    \hline
  7 & 0.1 & 4.00 & 2.00 & 6.50 & 128 &   7 & 2.00 & 1.00 & 3.00 \\ 
  \textbf{8} &\textbf{0.3} & \textbf{8.00} & \textbf{4.00 }& \textbf{16.00 }& \textbf{2048} &  \textbf{11} & \textbf{3.00} & \textbf{2.00} &\textbf{ 4.00 }\\ 
  9 & 0.5 & 8.00 & 4.00 & 16.00 & 768 &  10 & 3.00 & 2.00 & 4.00 \\ 
  10 & 0.7 & 4.00 & 2.00 & 8.00 & 256 &   8 & 2.00 & 1.00 & 3.00 \\ 
  11 & 0.9 & 1.00 & 1.00 & 2.00 &  16 &   4 & 0.00 & 0.00 & 1.00 \\ 
  12 & 1 & 1.00 & 1.00 & 2.00 &   8 &   3 & 0.00 & 0.00 & 1.00 \\ 
  \\ 
    \hline
    & ~ & ~ & ~ & ~ & \bf{$\phi_m = 0.5$} &  ~  & ~ & ~ & ~  \\
    \hline
  13 & 0.1 & 2.00 & 1.00 & 4.00 &  64 &   6 & 1.00 & 0.00 & 2.00 \\ 
  14 & 0.3 & 8.00 & 4.00 & 16.00 & 1024 &  10 & 3.00 & 2.00 & 4.00 \\ 
  \textbf{15} & \textbf{0.5} & \textbf{16.00} & \textbf{6.00} & \textbf{32.00} & \textbf{1024} &  \textbf{11} & \textbf{4.00} & \textbf{3.00} &\textbf{ 5.00} \\ 
  16 & 0.7 & 8.00 & 4.00 & 16.00 & 512 &   9 & 3.00 & 2.00 & 4.00 \\ 
  17 & 0.9 & 2.00 & 1.00 & 4.00 &  96 &   7 & 1.00 & 0.00 & 2.00 \\ 
  18 & 1 & 1.00 & 1.00 & 2.00 &  12 &   5 & 0.00 & 0.00 & 1.00 \\ 
  \\ 
    \hline
    & ~ & ~ & ~ & ~ & \bf{$\phi_m = 0.7$} &  ~  & ~ & ~ & ~  \\
    \hline
  19 & 0.1 & 2.00 & 1.00 & 2.00 &  16 &   4 & 1.00 & 0.00 & 1.00 \\ 
  20 & 0.3 & 4.00 & 2.00 & 8.00 & 192 &   8 & 2.00 & 1.00 & 3.00 \\ 
  21 & 0.5 & 8.00 & 4.00 & 16.00 & 384 &  10 & 3.00 & 2.00 & 4.00 \\ 
  \textbf{22} & \textbf{0.7} & \textbf{12.00} & \textbf{4.00} & \textbf{32.00} & \textbf{2304} &  \textbf{12} & \textbf{4.00} & \textbf{2.00} & \textbf{5.00} \\ 
  23 & 0.9 & 4.00 & 2.00 & 8.00 & 192 &   9 & 2.00 & 1.00 & 3.00 \\ 
  24 & 1 & 2.00 & 1.00 & 4.00 &  14 &   6 & 1.00 & 0.00 & 2.00 \\ 
  \\ 
    \hline
    & ~ & ~ & ~ & ~ & \bf{$\phi_m = 0.9$} &  ~  & ~ & ~ & ~  \\
    \hline
  25 & 0.1 & 1.00 & 1.00 & 1.00 &   8 &   3 & 0.00 & 0.00 & 0.00 \\ 
  26 & 0.3 & 1.00 & 1.00 & 2.00 &  16 &   5 & 0.00 & 0.00 & 1.00 \\ 
  27 & 0.5 & 2.00 & 1.00 & 4.00 &  96 &   7 & 1.00 & 0.00 & 2.00 \\ 
  28 & 0.7 & 4.00 & 2.00 & 8.00 &  96 &   8 & 2.00 & 1.00 & 3.00 \\ 
 \textbf{29} & \textbf{0.9 }& \textbf{8.00} & \textbf{4.00} & \textbf{16.00} & \textbf{280} &  \textbf{14} & \textbf{4.00} & \textbf{2.00} & \textbf{5.00} \\ 
  30 & 1 & 8.00 & 4.00 & 12.00 & 112 &  15 & 4.00 & 3.00 & 6.00 \\ 
  \\ 
    \hline
    & ~ & ~ & ~ & ~ & \bf{$\phi_m = 1.0$(Uniform)} &  ~  & ~ & ~ & ~  \\
    \hline
  31 & 0.1 & 1.00 & 1.00 & 1.00 &   4 &   2 & 0.00 & 0.00 & 0.00 \\ 
  32 & 0.3 & 1.00 & 1.00 & 2.00 &   8 &   3 & 0.00 & 0.00 & 1.00 \\ 
  33 & 0.5 & 1.00 & 1.00 & 2.00 &   9 &   5 & 0.00 & 0.00 & 1.00 \\ 
  34 & 0.7 & 2.00 & 1.00 & 4.00 &  30 &   7 & 1.00 & 0.00 & 2.00 \\ 
  35 & 0.9 & 7.00 & 4.00 & 12.00 & 116 &  15 & 5.00 & 3.00 & 7.00 \\ 
  \textbf{36} & \textbf{1} & \textbf{82.00} &\textbf{ 61.00} & \textbf{112.00} & \textbf{794} &  \textbf{53} & \textbf{33.00} & \textbf{29.00} & \textbf{37.00 }\\ 
   \hline
\end{tabular}
\caption{Summary statistics for the size of the lattice ($\mathcal{L}$) and the rotation poset ($\mathcal{R}$) in Mallows model with different dispersion parameters for men ($\phi_m$) and women($\phi_w$). Q1, Q2, Q3 indicate 25th, 50th (median), 75th percentiles. }
\label{table:summary_phi}
\end{table*}

\subsubsection*{\textbf{Rotation Poset}}
    The rotation poset size follows the same pattern as the lattice size: it has the maximum size in case of zero correlation disparity and decreases with the increase in the disparity. The rotation poset is much smaller under the Mallows model (all cases) compared to the Uniform. For $n=150$, the median rotation poset size in the Uniform model ($33$) is substantially larger than in the Mallows model (maximum median equals $4$ rotations) (Table \ref{table:rot_pos}). For $n=300$, the difference is even more pronounced: $60$ rotations vs $8$. For a general case, if rotation poset $\mathcal{R}$ has $r$ rotations and height $h$, the maximum number of downsets in a poset, $d(P)$, (corresponding to maximum stable lattice size) can be counted as $d(P)=2^{r-h}(h+1)$ \citep{campo2018exponential}. We see an interesting configuration of the rotation posets associated with large lattices: a decreased height (the size of the maximal chain), and an increased width (the size of the largest anti-chain) \citep{campo2018exponential}. For example, for $n=150$ the maximum stable lattice under the Mallows model has $2304$ matchings, only $12$ corresponding rotations and rotation poset height equals $2$ (Table \ref{table:summary_phi}). Under the Uniform model the maximum stable lattice was smaller but had many more rotations: $626$ matchings and $53$ rotations (the rotation poset height was equal to $4$). 

\section{Fine-grained asymmetric correlation case} \label{app:fine-grained}
In the part of the graph, where women's preferences have a lower dispersion parameter $\phi_w<\phi_m$, the distribution of $S_W$-scores of a women-optimal matching lies above the distribution of its $S_M$-scores: the distributions have no overlap for  $\phi_w \leq 0.96$ (\cref{fig:M_W_0.99}).
\begin{figure}
    \centering
    \includegraphics[width=0.6\columnwidth]{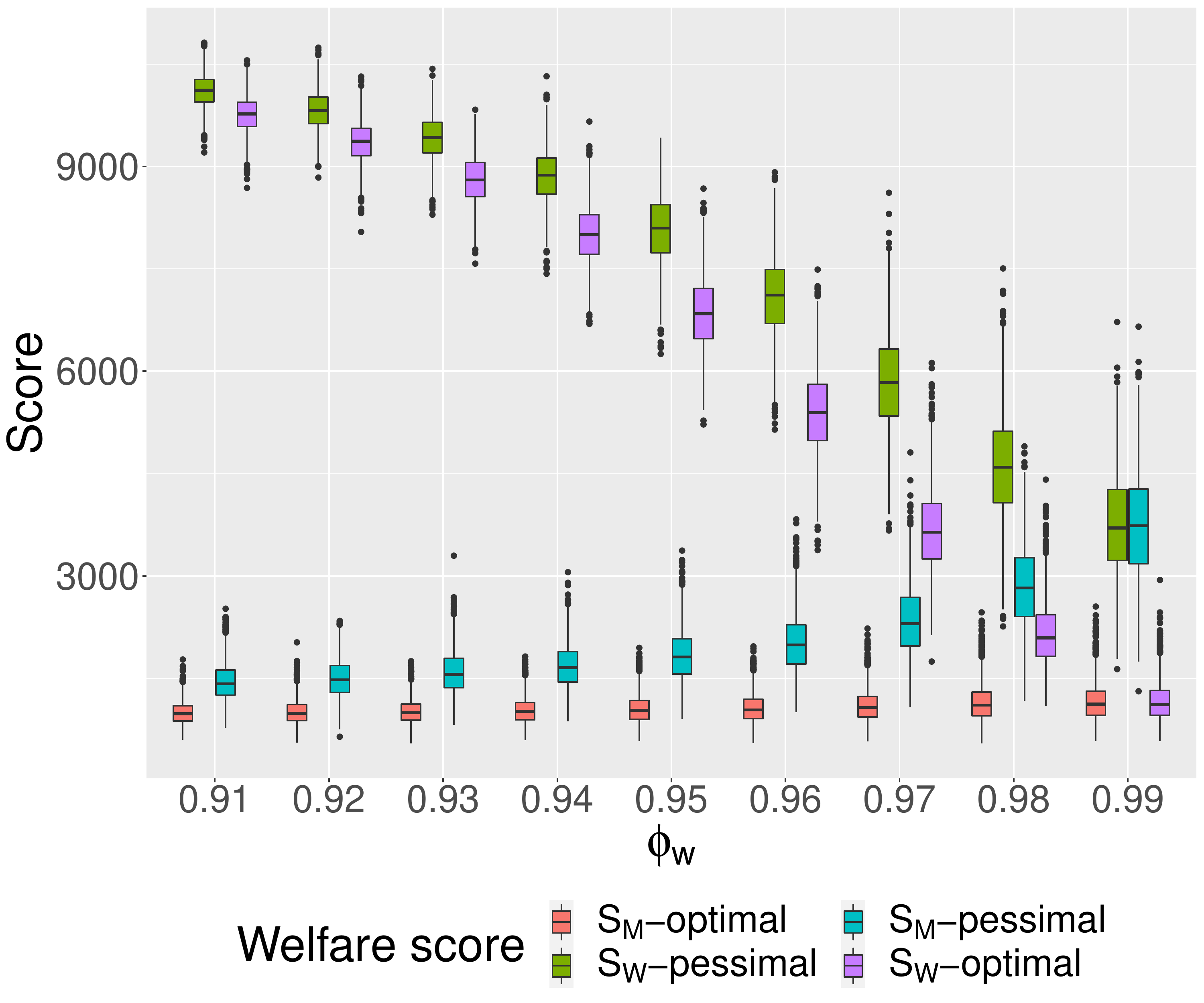}
    \caption{Illustrating $S_M$, $S_W$ scores under the Mallows model with $\phi_m = 0.99$ and varying $\phi_w$ based on 1000 samples.
    }
    \label{fig:M_W_0.99}
\end{figure}

\section{Egalitarian cost under the Mallows model} \label{sec:egal}
To assess the total welfare of \DA{} outcome in the presence of correlation, we looked into its egalitarian cost. We found that introducing correlation in preferences of any side translates into a high total dissatisfaction of all agents by the \DA{} outcome. 

Given a preference profile $\>$, the \emph{egalitarian cost} of a stable matching is the total ranking over all men and women \cite{irving1987efficient}. Formally,

\begin{equation}
c'(\mu) = S_M(\mu) + S_W(\mu)
\end{equation}

The smaller the egalitarian cost, the higher the welfare of all involved agents.

According to our experiments, \DA{} algorithm is a reasonable alternative to the exhaustive search algorithm for finding a sex-equal stable matching, when the preferences come from the Mallows model. Its output seems to be either a sex-equal solution (in asymmetric correlation markets) or a matching close to sex-equal one in terms of the cost (in symmetric correlation markets). Although \DA{} achieves such a high level of sex-equality fairness in the Mallows, the overall satisfaction of agents by its output remains low in experiments (\Cref{fig:Ega_heat}). The egalitarian cost is high for all $\phi$ combinations under the Mallows distribution: the smallest median score is $16636$ for $\phi_m = \phi_w = 0.9$, which for $n=150$ corresponds to having a partner with $55$ rank position on average. In contrast, the \DA{} applied to the uniform preferences have a relatively high total satisfaction, while having a great disparity between the welfare of men and women according to our experiments (\Cref{fig:cost_se}).  

\begin{figure}[h!]
    \centering
    \includegraphics[width = 0.8\columnwidth]{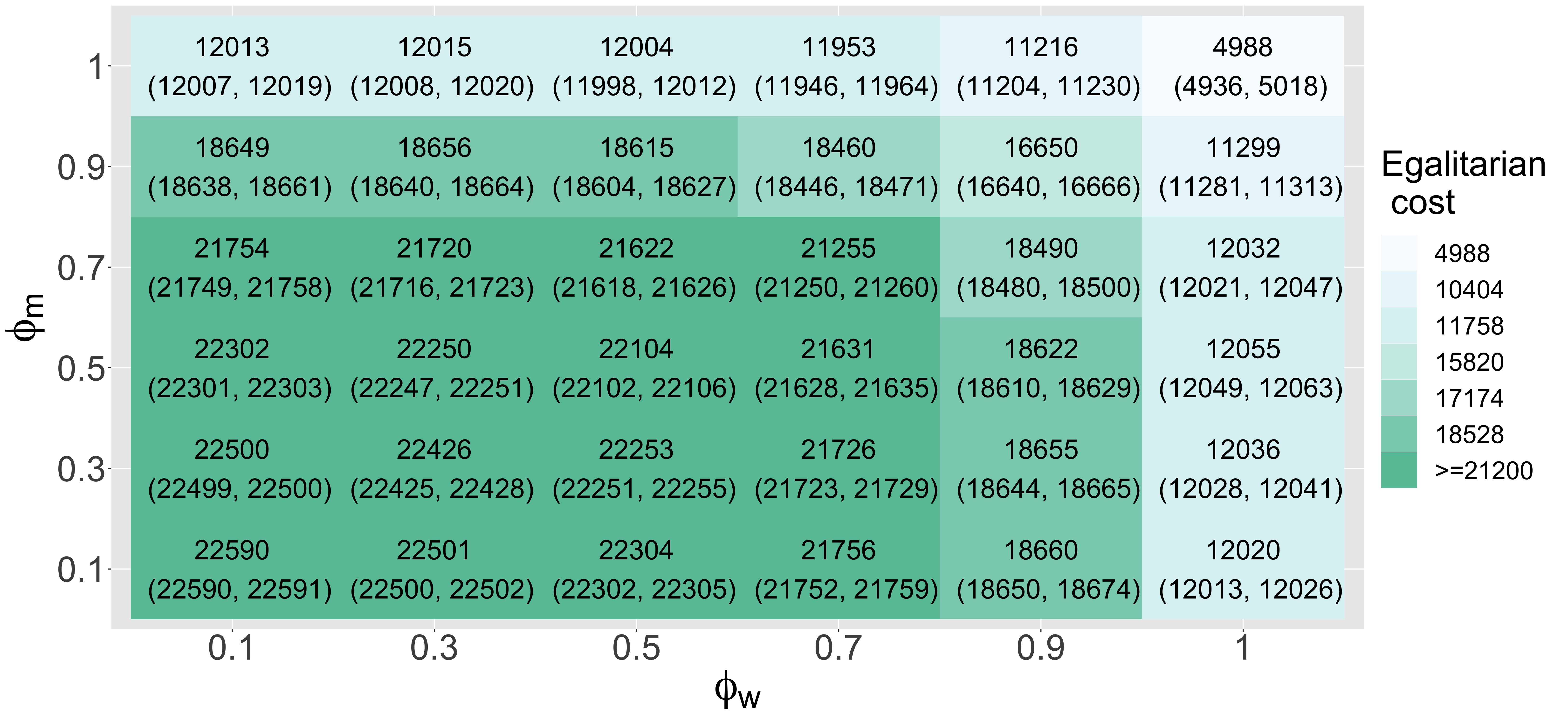}
    \caption{The egalitarian cost of a men-optimal matching under the Mallows model with various $\phi_m$, $\phi_w$ and $n=150$. Confidence intervals are given in parentheses.}
    \label{fig:Ega_heat}
\end{figure}

\section{Algorithm performance in the Uniform model}\label{app:alg_uni}
Under the Uniform model, the outcomes of the \DA{} algorithm have high sex-equality costs, as expected, since men- and women-optimal matchings have high sex-equality cost \cite{pittel1989average}. Other algorithms performed better: iBILS and PowerBalance showed the best performance and in most cases were able to find a sex-equal stable matching.
\begin{figure}[h!]
\centering
    \includegraphics[width = 0.8\columnwidth]{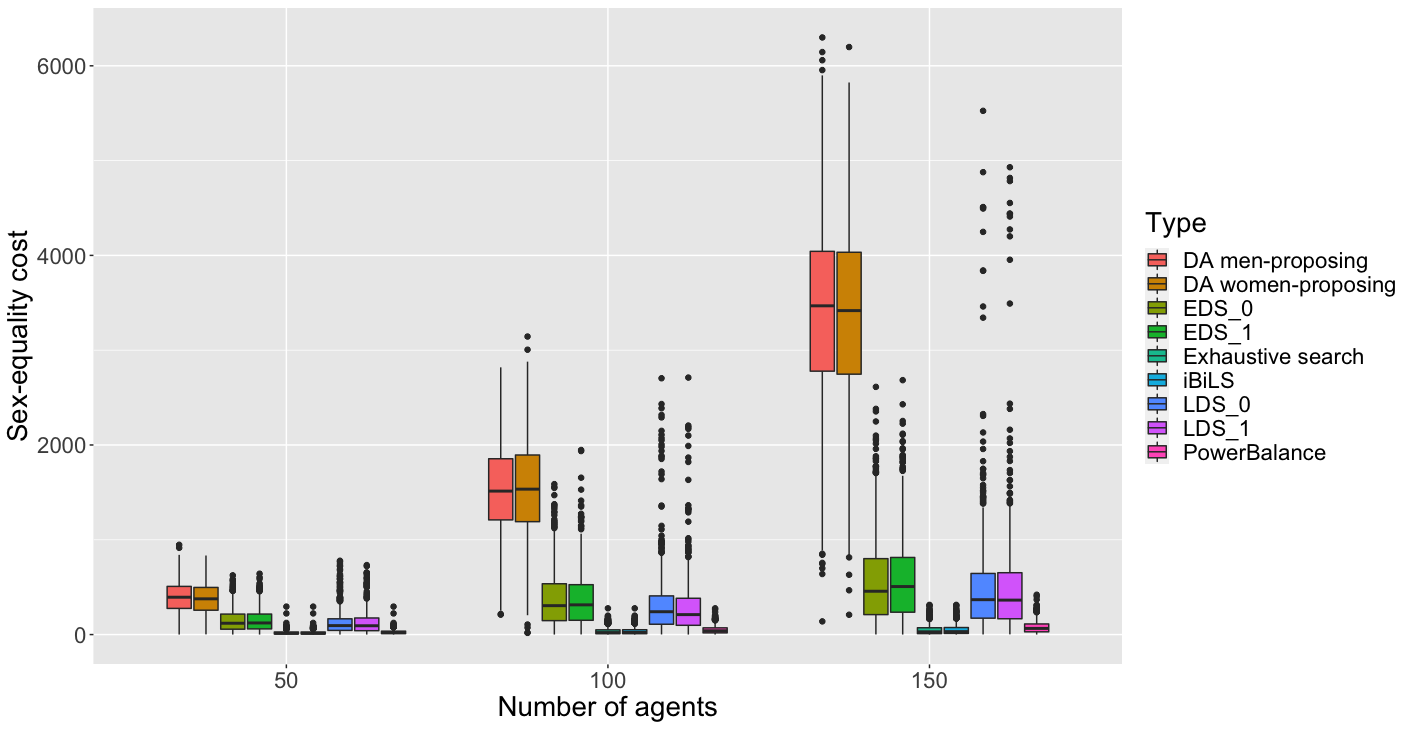}
    \captionof{figure}{Uniform model: comparing the sex-equality cost of various matching algorithms. Subscripts 0, indicate that men start proposing first, 1 that women.}
    \label{fig:comparison}
\end{figure}

\section{Fairness beyond the Mallows model}

\subsection{Other Correlated Models}
We also analyzed the fairness of other correlated preference models known in the literature. Namely, we analyze the known results for some probabilistic models with hard-tier sets as well as those that consider correlation coefficients on preferences. It is unclear whether these models can be thought of as Mallows models (or Mallows mixture). 
Yet, similar to the Mallows model correlation intensity between the preferences of two sides might be a reasonable indicator in predicting the sex-equality cost of stable matchings.

\subsubsection*{\textbf{Asymmetric correlation}} The \textit{Discrete Uniform model} from \cite{tziavelis20} is a probabilistic hard-tier preference model in which a subset of agents is always ranked higher than the rest of the set---referred to as \textit{hot} and \textit{cold} sets. Within a set, partners are sampled uniformly. According to experiments of \cite{tziavelis20} when men have uniform preferences on women, and women have discrete uniform preferences on men, the women-proposing \DA{} has one of the lowest sex-equality cost and drastically outperforms many procedurally fair algorithms. In this case, women's preferences are more correlated compared to men, so the women-optimal matching might become a sex-equal matching (or gives a good approximation of its sex-equality cost), coinciding with our findings for the Mallows model in asymmetric correlation markets.

\subsubsection*{\textbf{Symmetric correlation}} \citet{celik2007marriage} proposes a correlation coefficient of preferences $\rho$. The coefficient computes the sum of agents' average ranks in preferences of men (or women), normalized to take values between 0 (perfectly uncorrelated preferences) and 1 (identical preferences). The authors used the correlation coefficients of men $\rho_m$ and women $\rho_w$ as a predictor of men's and women's satisfaction by \DA{} outcome. To generate instances with different correlation intensities they sampled from a hard-tier model while varying the number of tiers. They showed that a higher correlation in preferences of women sharply reduces their satisfaction by a men-optimal matching while making the other side slightly better off. These insights are consistent with our observations (e.g. in \cref{fig:M_W_0.9} $S_W$-pessimal rises, and $S_M$-optimal moderately decreases when $\phi_w$ decreases). 
\newpage


\end{document}